\documentclass{article}

\usepackage[final]{neurips_2022}

\usepackage[utf8]{inputenc} 
\usepackage[T1]{fontenc}    
\usepackage{url}            
\usepackage{booktabs}       
\usepackage{amsfonts}       
\usepackage{nicefrac}       
\usepackage{microtype}      
\usepackage{xcolor}         

\usepackage[algo2e, ruled, vlined]{algorithm2e}
\usepackage{amsmath}
\usepackage{amsfonts}
\usepackage{graphicx}
\usepackage{amsthm}
\usepackage{dsfont}
\usepackage{bbm}

\usepackage[colorinlistoftodos]{todonotes}
\usepackage[colorlinks=true, allcolors=blue]{hyperref}

\theoremstyle{theorem}
\newtheorem{theorem}{Theorem}
\newtheorem{fact}{Fact}
\newtheorem{claim}{Claim}
\newtheorem{lemma}{Lemma}
\newtheorem{corollary}{Corollary}

\newtheorem{remark}{Remark}
\newenvironment{reminder}[1]{\smallskip
	\noindent {\bf Reminder of #1.}\em}{\smallskip}

\theoremstyle{definition}
\newtheorem{definition}{Definition}

\newenvironment{proofof}[1]{\begin{proof}[{\textit{Proof of #1}}]}{\end{proof}}

\newcommand{\Lower}{\mathsf{Lower}}
\newcommand{\Upper}{\mathsf{Upper}}
\newcommand{\COMP}{\mathrm{COMP}}
\newcommand{\calA}{\mathcal{A}}
\newcommand{\calB}{{\mathcal{B}}}

\newcommand{\calE}{{\mathcal{E}}}

\newcommand{\calM}{{\mathcal{M}}}
\newcommand{\calN}{{\mathcal{N}}}
\newcommand{\calX}{{\mathcal{X}}}
\newcommand{\calY}{{\mathcal{Y}}}
\newcommand{\calZ}{{\mathcal{Z}}}
\newcommand{\eps}{\varepsilon}

\newcommand{\calS}{\mathcal{S}}

\newcommand{\IT}{\mathbf{IT}}
\newcommand{\RR}{\mathrm{RR}}
\newcommand{\Output}{\mathbf{Output}}
\newcommand{\supp}{\mathrm{supp}}
\DeclareMathOperator*{\Ex}{\mathbb{E}}

\title{Composition Theorems for Interactive Differential Privacy}

\author{%
Xin Lyu \\
Department of Electrical Engineering and Computer Science \\
University of California, Berkeley \\
Berkeley, CA, 94720 \\
\text{xinlyu@berkeley.edu} \\
}

\begin{document}

\maketitle

\begin{abstract}
    An interactive mechanism is an algorithm that stores a data set and answers adaptively chosen queries to it. The mechanism is called differentially private, if any adversary cannot distinguish whether a specific individual is in the data set by interacting with the mechanism. We study composition properties of differential privacy in concurrent compositions. In this setting, an adversary interacts with $k$ interactive mechanisms in parallel and can interleave its queries to the mechanisms arbitrarily. Previously, \citet{DBLP:conf/tcc/VadhanW21} proved an optimal concurrent composition theorem for pure-differential privacy. We significantly generalize and extend their results. Namely, we prove optimal parallel composition properties for several major notions of differential privacy in the literature, including approximate DP, R\'enyi DP, and zero-concentrated DP. Our results demonstrate that the adversary gains no advantage by interleaving its queries to independently running mechanisms. Hence, interactivity is a feature that differential privacy grants us for free.
    
    Concurrently and independently of our work, \citet{VadhanZ2022} proved an optimal concurrent composition theorem for $f$-DP \citep{Dong2022GaussianDP}, which implies our result for the approximate DP case.
\end{abstract}

\section{Introduction}

By now, differential privacy \citep{DBLP:conf/tcc/DworkMNS06} has been widely accepted as a standard framework for protecting individual privacy when performing data analysis on data sets that may contain sensitive information of individuals (see, e.g., the surveys by \citet{DBLP:journals/fttcs/DworkR14, DBLP:books/sp/17/Vadhan17}).

Let $\calM$ be an algorithm that runs on a data set $x$ and calculates some information about it. Roughly speaking, $\calM$ is called differentially private, if the output distribution of $\calA$ remains nearly identical when we arbitrarily modify a single entry in $x$.

One essential feature of differential privacy is its \emph{composability}. Composition captures the scenario where a data analyst runs $k$ differentially private algorithms sequentially, and releases the results afterward. Typically, a composition theorem has the following form: if each of the $k$ algorithms satisfies differential privacy, then the analyst's output is still differentially private with moderately degraded privacy parameters.


Composition theorems are important for at least two reasons. First, we might want to perform computation tasks on the same data set multiple times and still have reasonable control over the privacy loss. In this case, composition theorems reveal how the privacy guarantee degrades over time. More importantly, composition theorems allow us to build more complex and powerful differentially-private algorithms from simple primitives, and argue the privacy guarantee of the combined algorithm in a straightforward way.

There is a rich literature concerning the composition property of differential privacy (see, e.g., \citet{DBLP:conf/eurocrypt/DworkKMMN06, DBLP:conf/focs/DworkRV10, DBLP:conf/icml/KairouzOV15, DBLP:journals/toc/MurtaghV18, DBLP:journals/siamcomp/BassilyNSSSU21}). However, most existing composition theorems only consider the scenario where an analyst runs several private algorithms \emph{sequentially}. That is, the analyst will only move on to the next algorithm after finishing their computation with the previous one. In contrast, many fundamental primitives in differential privacy are interactive in nature, such as the sparse vector technique \citep{DBLP:conf/stoc/DworkNRRV09-sparse_vector, DBLP:conf/stoc/RothR10-sparse_vector} and private multiplicative weight updates \citep{DBLP:conf/focs/HardtR10-privateMWU}. Hence, the interactivity issue appears to be a significant limitation of current composition theorems. Namely, the data analyst may want to communicate with several interactive mechanisms \emph{concurrently}, and interleave its queries to the mechanisms arbitrarily. A sequential composition theorem completely fails to capture this scenario. Also, in practice, deployments of DP algorithms often demand a better understanding of concurrent compositions of interactive mechanisms \citep{OpenDP}.

Recently, \citet{DBLP:conf/tcc/VadhanW21} initiated a study of concurrent compositions and proved an optimal concurrent composition theorem for pure differential privacy. In this work, we significantly advance this research direction by proving optimal concurrent composition theorems for several popular notions of differential privacy, including approximate DP, R\'enyi DP, zero-concentrated DP and truncated concentrated DP.

\subsection{Setup}

Before we continue, we set up necessary pieces of notation. We use $\calX$ and $\calY$ to denote the domain of query messages and responses, respectively. We assume that both $\calX$ and $\calY$ are finite sets. This assumption is for easing some mathematical manipulation and is not restrictive: all practical applications of differential privacy have finite input and output spaces anyway.

For a set $S$, denote by $\Delta(S)$ the set of all possible distributions supported over $S$. We define interactive systems below.

\begin{definition}[Interactive system]\label{def:system}
An interactive system is a (randomized) algorithm $\calM\colon (\calX\times \calY)^* \times \calX \to \Delta(\calY)$. The input to $\calM$ is an interaction history $(x_1,y_1),(x_2,y_2),\dots, (x_t, y_t)\in (\calX\times \calY)^t$ together with a query $x_{t+1}$. The output of $\calM$ is denoted by $y_{t+1} \sim \calM((x_i,y_i)_{i\in [t]}, x_{t+1})$.
\end{definition}

A technicality worth mentioning is that due to the internal memory and randomness of an interactive system $\calM$, the response of $\calM$ to the $(t+1)$-th query might be correlated with its responses to previous queries. Although the internal randomness of $\calM$ is not explicitly stated as a parameter, Definition~\ref{def:system} captures this correlation by requiring that each query $x_{t+1}$ to $\calM$ is attached with the interaction history $(x_1,y_1),\dots, (x_t, y_t)$. This history is sufficient for determining the conditional distribution of the response $\calM((x_i,y_i)_{i\in [t]}, x_{t+1})$ without specifying the internal randomness and memory.

We make a distinction between mechanisms and systems. By ``mechanism'' we mean a differentially private algorithm $\calM$ that holds a sensitive input $d$ and answers queries about it. When applied to a concrete input $d$, $\calM$ induces an interactive system, denoted by $\calM^d$. According to the definition of differential privacy, studying the privacy of a mechanism boils down to studying the pair of systems $(\calM^{d},\calM^{d'})$ induced by running $\calM$ on every pair of neighboring inputs $(d,d')$. For brevity, we usually assume W.L.O.G. that the input only consists of a single bit $b\in \{0,1\}$, and we compare the two systems $\calM^0,\calM^1$ induced by $\calM$.

\noindent\textbf{Concurrent composition.} We define concurrent composition of interactive systems. Suppose $\calM_1,\calM_2,\dots, \calM_k$ are $k$ systems. The concurrent composition of them is an interactive system $\COMP(\calM_1\dots \calM_k)$ with query domain $[k]\times \calX$ and response domain $\calY$. An adversary is a (possibly randomized) query algorithm $\calA : ([k]\times \calX \times \calY)^*\to \Delta([k]\times \calX)$. The interaction between $\calA$ and $\COMP(\calM_i)$ is a stochastic process that runs as follows. $\calA$ first\footnote{We assume it is always the adversary who sends the first message. This is without loss of generality: we can let the first message sent from the adversary to each system be an ``Initiliazation'' query. Having received the initialization query, the system returns either a starting message or simply a ``SUCCESS'' symbol.} computes a pair $(i_1,x_1)\in [k]\times \calX$, sends a query $x_1$ to $\calM_{i_1}$ and gets the response $y_1$. In the $t$-th step, $\calA$ calculates the next pair $(i_t,x_t)$ based on the history, sends the $t$-th query $x_t$ to $\calM_{i_t}$ and receives $y_t$. There is no communication or interaction between the interactive systems. Each system $\calM_i$ can only see its own interaction with $\calA$. Let $\IT(\calA : \calM_1,\dots, \calM_k)$ denote the random variable recording the transcript of the interaction.

In the special case $k = 1$, there is only one system $\calM$ and the adversary is interacting with it. We define approximate differential privacy for interactive mechanisms in this case.

\begin{definition}[Indistinguishability and $(\eps,\delta)$-DP]
Two interactive systems $\calM^0, \calM^1$ are called $(\eps, \delta)$-indistinguishable, if for every $b\in \{0,1\}$, every adversary $\calA$ and every collection of transcripts $S \subseteq \{ (x_i, y_i)_{i\in [T]} \}$, it holds that
\begin{align}
\Pr[\IT(\calA: \calM^{b}) \in S] \le e^{\eps} \Pr[\IT(\calA: \calM^{1-b})\in S] + \delta. \label{eq:indistinguishable-def}
\end{align}
Let $\calM$ be an interactive mechanism. $\calM$ is called $(\eps,\delta)$-approximate differentially private (or $(\varepsilon,\delta)$-DP for short), if for every two neighboring data sets $d$ and $d'$, the systems $\calM^d$ and $\calM^{d'}$ are $(\eps,\delta)$-indistinguishable.
\end{definition}

\subsection{Differential Privacy in Concurrent Compositions}

We study the privacy guarantee under concurrent compositions. Let $\calM^b_1,\dots, \calM^b_k$ be $k$ interactive mechanisms, each satisfying $(\eps,\delta)$-DP. Consider their concurrent composition $\COMP(\calM^b_1\dots\calM^b_k)$. We want to find out the smallest parameters $\eps',\delta'$ such that $\COMP(\calM^b_i)$ satisfies $(\eps',\delta')$-DP. In the sequential composition, the adversary $\calA$ interacts with $\calM_i$'s in order and cannot interleave its queries. In this case, it is known by the advanced composition theorem \citep{DBLP:conf/focs/DworkRV10} that $\IT(\calA : \calM^0_1,\dots,\calM^0_k)$ and $\IT(\calA : \calM^1_1,\dots, \calM^1_k)$ are $(O(\sqrt{k \log(1/\delta')}\eps),k\delta + \delta')$-indistinguishable.

However, in general, the adversary can interleave its queries arbitrarily, and the differential privacy guarantee warranted by $\COMP(\calM_i)$ is less clear. \citet{DBLP:conf/tcc/VadhanW21} were the first to formally study this question. They showed that in the special case $\delta = 0$, an optimal composition holds for $\COMP(\calM_i)$. That is, if we can prove an $(\eps',\delta')$ upper bound on the privacy parameter for sequential compositions of $\calM^b_1,\dots, \calM^b_k$, then the concurrent composition $\COMP(\calM_i)$ also enjoys the same $(\eps',\delta')$-DP. 

\citet{DBLP:conf/tcc/VadhanW21} also considered the case $\delta > 0$ (i.e., approximate DP). However, for this case, they only showed an upper bound on $\eps',\delta'$ that is inferior to the basic composition in the sequential setting. It was asked as an open question in \cite{DBLP:conf/tcc/VadhanW21} whether the optimal composition theorem for approximate DP still holds in the concurrent composition.

Besides pure and approximate DP, there are also other notions of differential privacy that are extensively studied in the literature. A non-exhaustive list includes R\'enyi DP \citep{DBLP:conf/csfw/Mironov17-renyi_dp}, concentrated DP \citep{DBLP:journals/corr/DworkR16, DBLP:conf/tcc/BunS16,DBLP:conf/stoc/BunDRS18}, Gaussian DP and $f$-DP \citep{Dong2022GaussianDP} etc. Compared with the standard notion of $(\eps,\delta)$-approximate DP, these variants of DP either allow for a simplified analysis of private algorithms or give sharper bounds of privacy guarantee. In the sequential composition, the composition property of these variants has been well understood. It is also interesting to extend these composition theorems to the concurrent composition, thereby expanding the potential applicability of these DP notions.

\section{Our Results}

In this work, we give an affirmative answer to the open question mentioned above. Moreover, our result confirms that several major differential privacy definitions in the literature enjoy the same composition guarantee in the concurrent composition, just as they do in the sequential composition.

\noindent\textbf{Approximate differential privacy.} $(\eps,\delta)$-DP is arguably the most widely studied notion of differential privacy and is deemed the ``standard'' definition of DP. As our first main result, we show an optimal concurrent composition theorem for approximate DP.

\begin{theorem}\label{theo:approximate}
Let $\calM_1,\dots, \calM_k$ be $k$ interactive mechanisms that run on the same data set. Suppose that each mechanism $\calM_i$ satisfies $(\eps_i, \delta_i)$-DP. Then $\COMP(\calM_1\dots\calM_k)$ is $(\eps', \delta')$-DP, where $\eps', \delta'$ are given by the optimal (sequential) composition theorem \citep{DBLP:conf/icml/KairouzOV15,DBLP:journals/toc/MurtaghV18}.

In particular, when the privacy parameter for each mechanism is the same $(\eps, \delta)$, their concurrent composition satisfies $O(\sqrt{k\log(1/\delta')}\eps, \delta' + k\delta)$-DP for all $\delta' \in (0, 1)$.
\end{theorem}
\noindent\textbf{R\'enyi differential privacy.} R\'enyi differential privacy was first defined by \citet{DBLP:conf/csfw/Mironov17-renyi_dp}. We recall its definition.

\begin{definition}[R\'enyi divergence and differential privacy]\label{def:renyi}
Let $P, Q$ be two distributions supported over $\calX$. For each $\alpha > 1$, define the R\'enyi divergence of order $\alpha$ of $P$ from $Q$ as
\[
D_{\alpha}(P\| Q) := \frac{1}{\alpha - 1} \log\left( \Ex_{x\sim P}\left[ \left(\frac{P(x)}{Q(x)} \right)^{\alpha - 1} \right] \right).
\]
Two interactive systems are called $(\alpha, \eps)$-R\'enyi close, if for every adversary $\calA$ and every $b\in \{0,1\}$, it holds that
\[
D_{\alpha}( \IT(\calA : \calM^b) \| \IT(\calA:\calM^{1-b})) \le \eps.
\]
Let $\calM$ be a mechanism. $\calM$ is called $(\alpha,\eps)$-R\'enyi differentially private (or $(\alpha,\eps)$-RDP for short), if for every two neighboring data sets $d$ and $d'$, the systems $\calM^d$ and $\calM^{d'}$ are $(\alpha,\eps)$-R\'enyi close.
\end{definition}
A main advantage of R\'enyi DP is that it has a natural and simple composition. In the sequential setting, it is known that if two mechanisms $\calM_1,\calM_2$ are $(\alpha,\eps_1)$ and $(\alpha,\eps_2)$-RDP, respectively, then the composition of $\calM_1$ and $\calM_2$ is $(\alpha,\eps_1 + \eps_2)$-RDP. Our next theorem generalizes this result to the concurrent composition setting.
\begin{theorem}\label{theo:renyi}
Let $\calM_1,\dots, \calM_k$ be $k$ interactive mechanisms that run on the same data set. Suppose that each mechanism $\calM_i$ is $(\alpha, \eps_i)$-RDP. Then $\COMP(\calM_1\dots\calM_k)$ is $(\alpha, \sum_{i=1}^{k} \eps_i)$-RDP.
\end{theorem}
One implication of Theorem~\ref{theo:renyi} is that the zero-concentrated differential privacy by \citet{DBLP:conf/tcc/BunS16} and the truncated concentrated differential privacy by \cite{DBLP:conf/stoc/BunDRS18} also compose nicely under the concurrent composition. We state the corollary below, and prove it in Appendix~\ref{append:concentrate} for completeness.

\begin{corollary}\label{coro:CDP}
Let $\calM_1,\dots, \calM_k$ are $k$ interactive mechanisms that run on the same data set. Suppose that each mechanism $\calM_i$ is $\eta_i$-zCDP (resp.~$(\rho_i, \omega)$-tCDP). Then $\COMP(\calM_1\dots\calM_k)$ is $(\sum_i{\eta_i})$-zCDP (resp.~$(\sum_i \rho_i, \omega)$-tCDP).
\end{corollary}


Theorem~\ref{theo:approximate}, \ref{theo:renyi} and Corollary~\ref{coro:CDP} provide compelling evidence that the adversary gains no advantage by interleaving its queries to independently running mechanisms. Consequently, interactivity can be viewed as a feature that differential privacy grants us for free.

\noindent\textbf{Concurrent and independent work.} Concurrently to our work, a recent work by \citet{VadhanZ2022} proves an optimal concurrent composition theorem for $f$-DP \citep{Dong2022GaussianDP}. By the standard connection, their result implies the optimal concurrent composition theorem for approximate DP. However, our techniques are very different than theirs. Their result is stronger, as it is known that approxiamte-DP can be seen as a special case of $f$-DP \citep{Dong2022GaussianDP}. However, our proof for approximate DP is more elementary: we do not need to work through $f$-DP as their proof does. Furthermore, our proof comes with several interesting technical ingredients that might be of independent interests. This includes a structural result for interactive mechanisms (Theorem~\ref{theo:decomposition}), as well as a dual perspective to reason about R\'enyi divergences (Lemma~\ref{lemma:renyi-charact}).

\section{Implications of Our Results}

In this section, we discuss implications of our results, and demonstrate how they offer more than the sequential composition theorems.

\noindent\textbf{Designing new algorithms.} The optimal concurrent composition theorem makes it possible to design new differentially private algorithm that involves running several building blocks concurrently. As one motivating example, consider the Sparse Vector Technique (SVT). The standard SVT (as in \citet{DBLP:journals/fttcs/DworkR14}) and its variants have been studied extensively in the literature. In particular, it was observed by \cite{DBLP:journals/pvldb/LyuSL17, DBLP:conf/nips/ZhuW20} that one can add noise to the threshold only \emph{once}, and then use the noisy threshold to answer $c > 1$ ``meaningful'' queries (namely, after reporting each meaningful query, the SVT algorithm does NOT refresh the noisy threshold). It was argued in \citep{DBLP:journals/pvldb/LyuSL17,DBLP:conf/nips/ZhuW20} that this variant of SVT can offer a higher accuracy while consuming the same amount of privacy budget, both theoretically and empirically.

However, this variant of SVT has received relatively less attention in literature. One reason might be that it is unclear what happens if we compose this SVT with other mechanisms. In particular, the standard SVT refreshes its threshold after answering each ``meaningful'' query, which allows one to decompose the algorithm into $c$ pieces of smaller SVT algorithms, and then compose with other mechanisms via the sequential computation. In contrast, the variants by \cite{DBLP:journals/pvldb/LyuSL17, DBLP:conf/nips/ZhuW20} work by answering each ``meaningful query'' using the \emph{same} noisy threshold, which do not seem to admit such a decomposition. This makes this variant less appealing: in most applications, people want to use SVT as a supporting subroutine for other algorithms. Therefore, it is crucial to understand the (concurrent) composition behavior of SVT with other mechanisms.

Now, with the new concurrent composition theorem, we can plug this variant of SVT in any algorithm, and argue the privacy guarantee of the whole computation by black-box applying Theorems~\ref{theo:approximate} and \ref{theo:renyi} (depending on whether we are working with $(\eps,\delta)$-DP or RDP). To illustrate the idea, in Appendix~\ref{append:example}, we apply Theorem~\ref{theo:approximate} to analyze a simple algorithm: private ``Guess-and-Check'' with the aforementioned variant of SVT as a subroutine. We hope our example can motivate people to design more powerful algorithms by concurrently composing simple building blocks.

\noindent\textbf{Practical Implication.} Besides the theoretical interests, our theorem has implications for practical deployments of interactive DP mechanisms. For one example, suppose there is a data center that holds the private information of individuals and offers data analysts access to the database (interactively and differentially-privately). Without knowing the concurrent composition theorem, it might be possible that some $k>1$ analysts can collude by coordinating their (interactive) queries to the database and extracting much more sensitive information. Our result refutes the possibility of such an attack. In particular, suppose each data analyst has only an $(\eps,\delta)$-DP amount of privacy ``quota''. Then, even if they collude and spend their privacy budget in whatever way, their computation result is still $(O(\sqrt{k\log(1/\delta')}\eps),k\delta+\delta')$-DP with respect to the private database.

\section{Proof of Main Results}\label{sec:proof}

In this section, we show the proof of our results. We start with a very brief proof overview. We prove Theorem~\ref{theo:approximate} by a reduction to the sequential composition of $k$ (approximate) randomized response mechanisms. This generalizes the idea developed by \citet{DBLP:conf/tcc/VadhanW21}. To prove Theorem~\ref{theo:renyi}, we take a completely different approach, and our technique offers new tools to analyze R\'enyi DP. Namely, we propose an alternative characterization of R\'enyi divergence (Lemma~\ref{lemma:renyi-charact}), which allows for a fine-grained account of the privacy loss in the complex interaction involving multiple mechanisms. The characterization of R\'enyi divergence might find itself useful in other applications.

\noindent\textbf{Notation.} Let $P, Q$ be two distributions supported over $X$. For a real $\eta > 0$, we say that $P \ge \eta Q$, if for every $S\subseteq X$, it holds that 
\[
\Pr_{x\sim P}[x\in S] \ge \eta \Pr_{x\sim Q}[x\in S].
\]
Furthermore, we say $P\equiv Q$, if $P$ and $Q$ are identically distributed.

\subsection{Approximate Differential Privacy}\label{sec:proof-approximate}


To prove Theorem~\ref{theo:approximate}, we follow the approach by \citet{DBLP:conf/tcc/VadhanW21}, where they showed that one can simulate two $(\eps, 0)$-indistinguishable interactive systems by post-processing a randomized response mechanism. This simulation enables them to reduce the concurrent composition to a sequential composition, and the optimal composition theorem follows. It was asked as an open question in \citet{DBLP:conf/tcc/VadhanW21} whether the same simulation can be carried out for approximate DP. We answer this question affirmatively.

\noindent\textbf{Review of the Vadhan-Wang approach.} It would be instructive to review the proof by \citet{DBLP:conf/tcc/VadhanW21} first. Let $\calM^0,\calM^1$ be the pair of systems by running the private mechanism on a pair of neighboring data sets. The adversary $\calA$ interacts with $\calM^b$ for some $b\in \{0,1\}$ and wants to find out the value of $b$. An intuitive yet delicate fact due to \citet{DBLP:conf/tcc/VadhanW21} is that, if $\calM^0$ and $\calM^1$ are $(\eps, 0)$-indistinguishable, then there exist two systems $\calN^0,\calN^1$ such that, for every adversary $\calA$, the distribution of $\IT(\calA: \calM^{b})$ is identical to $\frac{e^\eps}{1+\eps} \IT(\calA : \calN^{b}) + \frac{1}{1+e^{\eps}} \IT(\calA:\calN^{1-b})$. This enables one to simulate the many-round interaction between $\calA$ and $\calM^b$ by running a one-round randomized response mechanism.

In more detail, let $\RR^b_{\eps}$ denote the standard randomized response mechanism, defined as follows. $\RR^b_{\eps}$ only accepts one query. On the query, $\RR^b_\eps$ ignores the query message, returns $b$ with probability $\frac{e^\eps}{1+e^{\eps}}$, and returns $1-b$ otherwise. We modify $\calA$ to a new adversary $\calA'$: $\calA'$ first sends a query to $\RR^b_{\eps}$ and receives a bit $b'$. Then $\calA'$ simulates the interaction between $\calA$ and $\calN^{b'}$, and outputs the transcript (i.e., $\IT(\calA : \calN^{b'})$). Let $\Output(\calA': \RR^b_{\eps})$ denote the output distribution of $\calA'$ when interacting with $\RR^b_{\eps}$. It is clear that
\[
\Output(\calA' : \RR^b_{\eps}) \equiv \frac{e^\eps}{1+\eps} \IT(\calA:\calN^{b}) + \frac{1}{1+e^{\eps}} \IT(\calA:\calN^{1-b}) \equiv  \IT(\calA : \calM^b).
\]
Therefore, $\calA'$ simulates the interaction between $\calA, \calM^b$ faithfully, by a single query to $\RR_{\eps}^b$.

Now, suppose the adversary $\calA$ is interacting with $k$ mechanisms $\calM^b_1,\dots, \calM^b_k$ in parallel. For each $i\in [k]$, assuming that $\calM^0_i$ and $\calM^1_i$ are $(\eps_i, 0)$-indistinguishable, there is a decomposition of $\calM^0_i, \calM^1_i$ by some $\calN^0_i$ and $\calN^1_i$. Again, we modify $\calA$ to a new mechanism $\calA'$. $\calA'$ first queries $\RR^b_{\eps_i}, i\in [k]$ in order, and receives $k$ bits $b'_1,\dots, b'_k$. Then $\calA'$ simulates the interaction between $\calA$ and $(\calN^{b'_i}_{i})_{i\in [k]}$ and outputs the transcript. One can show that
\begin{align}
\Output(\calA' : \RR^b_{\eps_1}, \dots, \RR^b_{\eps_k}) \equiv \IT(\calA : \calM^b_1,\dots, \calM^b_k). \label{eq:simulate-1}
\end{align}
Note that the left hand side of \eqref{eq:simulate-1} can be simulated by a sequential composition of $k$ randomized response mechanisms. Invoking the optimal composition theorem for sequential composition \citep{DBLP:conf/icml/KairouzOV15, DBLP:journals/toc/MurtaghV18} concludes the proof.

\noindent\textbf{Extension to approximate DP.} Now, if $\calM^0$ and $\calM^1$ are $(\eps,\delta)$-indistinguishable with $\delta > 0$, there might not be a nice decomposition of $\calM^b$ into $\frac{e^\eps}{1+e^{\eps}} \calN^b + \frac{1}{1+e^{\eps}}\calN^{1-b}$. Still, it is plausible to conjecture that there is a decomposition of $\calM^0, \calM^1$ with four systems $\calN^0,\calN^1,\calE^0,\calE^1$ such that for each $b\in \{0,1\}$,
\begin{align}
    \calM^b = \delta \calE^b + (1-\delta) \left(  \frac{e^\eps}{1+e^{\eps}} \calN^b + \frac{1}{1+e^{\eps}} \calN^{1-b} \right). \label{eq:conj-decompose}
\end{align}

Our main technical result in this subsection proves the existence of such a decomposition.

\begin{theorem}\label{theo:decomposition}
Two systems $\calM^0,\calM^1$ are $(\eps,\delta)$-indistinguishable, if and only if there are four systems $\calN^0,\calN^1,\calE^0,\calE^1$ satisfying the following: for every adversary $\calA$ and $b\in \{0,1\}$, it holds that
\begin{align}
\IT(\calA : \calM^b) \equiv \delta \IT(\calA :\calE^b) + (1-\delta) \left( \frac{e^\eps}{1+\eps} \IT(\calA:\calN^{b}) + \frac{1}{1+e^{\eps}} \IT(\calA:\calN^{1-b})  \right). \label{eq:conj-decompose-2}
\end{align}
\end{theorem}

Theorem~\ref{theo:decomposition} implies Theorem~\ref{theo:approximate} by a similar reduction to (approximate) random response. For completeness, we include a proof in Appendix~\ref{append:approxiate}. 

We prove Theorem~\ref{theo:decomposition} by establishing a series of lemmas. In the following, we state these lemmas and explain their intuition. We defer the formal proof to Appendix~\ref{append:approxiate}.

\begin{lemma}\label{lemma:decomposition}
Suppose $\calM^0,\calM^1$ are $(\eps,\delta)$-indistinguishable. There are two systems $\calE^0,\calE^1$ satisfying the following.
\begin{itemize}
    \item For every adversary $\calA$ and $b\in \{0,1\}$, it holds that $\IT(\calA : \calM^b) \ge \delta\cdot  \IT(\calA : \calE^b)$.
    \item For every adversary $\calA$, every set of transcripts $S\subseteq \{(x_i,y_i)_{i\in [T]}\}$ and $b\in \{0,1\}$, it holds that
    \[
    \begin{aligned}
    &  ~~~~ \Pr[\IT(\calA : \calM^b)\in S] - \delta \Pr[\IT(\calA:\calE^b)\in S]  \\
    & \le e^\eps \left( \Pr[\IT(\calA : \calM^{1-b} ) \in S] - \delta \Pr[\IT(\calA : \calE^{1-b})\in S] \right).
    \end{aligned}
    \]
\end{itemize}
\end{lemma}

Roughly, Lemma~\ref{lemma:decomposition} says that there are two systems $\calE^0,\calE^1$ that capture the low-probability ``bad behavior'' of $\calM^0,\calM^1$. It is the primary technical contribution of this subsection. We prove Lemma~\ref{lemma:decomposition} by explicitly constructing the two systems $\calE^0,\calE^1$. That is, we specify the probability density functions $\Pr[\calE^b((x_j,y_j)_{j<i}, x_i) = y_i]$ for $\calE^0,\calE^1$ step by step, in the increasing order of $i = 1,2,\dots, T$.



\begin{lemma}\label{lemma:system-subtraction}
Suppose $\calM, \calE$ are two systems such that for every adversary $\calA$, it holds that $\IT(\calA :\calM) \ge \delta \IT(\calA : \calE)$. Then there is a system $\calN$ such that for every adversary $\calA$, it holds that
\[
\IT(\calA : \calM) \equiv \delta \IT(\calA : \calE)  + (1-\delta) \IT(\calA : \calN).
\]
\end{lemma}

For intuition, suppose $P, Q$ are two distributions such that $P \ge \delta Q$. Then one can easily find a distribution $Q'$ such that $P \equiv \delta Q + (1-\delta) Q'$. The proof of Lemma~\ref{lemma:system-subtraction} extends this simple idea.


\begin{lemma}[\citet{DBLP:conf/tcc/VadhanW21}]\label{lemma:pure-decompose}
Suppose $\calN^0,\calN^1$ are $(\eps, 0)$-indistinguishable. Then there are two systems $\calN^{0'},\calN^{1'}$ such that for every adversary $\calA$, it holds that
\[
\IT(\calA : \calN^b) \equiv \frac{e^\eps}{1+e^\eps} \IT(\calA : \calN^{b'}) + \frac{1}{1 + e^\eps} \IT(\calA : \calN^{(1-b)'}).
\]
\end{lemma}

\noindent\textbf{Wrap-up.} We can conclude the proof for Theorem~\ref{theo:decomposition} now. The ``if'' direction is obvious: the existence of a decomposition satisfying \eqref{eq:conj-decompose-2} implies that $\calM^0,\calM^1$ are $(\eps,\delta)$-indistinguishable. For the other direction, we start by constructing $\calE^0,\calE^1$ using Lemma~\ref{lemma:decomposition}. Then we construct $\calN^{0},\calN^{1}$ by Lemma~\ref{lemma:system-subtraction}. Lemma~\ref{lemma:decomposition} and \ref{lemma:system-subtraction} together ensure that $\calN^{0}$ and $\calN^1$ are $(\eps, 0)$-indistinguishable, which enables us to invoke Lemma~\ref{lemma:pure-decompose} and decompose $\calN^0,\calN^1$ into $\calN^{0'},\calN^{1'}$.  $(\calE^0,\calE^1,\calN^{0'},\calN^{1'})$ forms the final decomposition. It is straightforward to verify that they satisfy \eqref{eq:conj-decompose-2}.

\subsection{R\'enyi Differential Privacy}\label{sec:proof-renyi}

Our result for R\'enyi differential privacy (Theorem~\ref{theo:renyi}) takes a completely different approach. 


\noindent\textbf{An intuition.} Let $\calM_1$ be an $(\alpha,\eps)$-R\'enyi DP mechanism. Intuitively, $(\alpha,\eps)$-R\'enyi DP means that $\calM_1$ has $\eps$ unit of privacy budget and can distribute it to $T$ queries. Viewing the privacy budget as a form of ``deposit'', we hope to argue that two or more independently running mechanisms spend their deposit independently, and an adversary cannot trigger any mechanism to spend more privacy budget than it holds by interacting with other mechanisms. 

However, unlike some intuitive and easy-to-measure resources such as time and energy, the notion of privacy loss looks somewhat illusive. Even worse, we need to reason about this elusive resource in a stochastic process consisting of interactions with multiple systems. It was not clear how one can quantify the privacy loss in such an interactive and complex process. Nonetheless, we manage to find a new approach to do so.

\noindent\textbf{An alternative characterization for R\'enyi DP.} We introduce the following  characterization of R\'enyi divergence based on H\"older's inequality and duality. That is, we prove

\begin{lemma}[An alternative characterization of R\'enyi divergence]\label{lemma:renyi-charact}
Suppose $P, Q$ are two distributions supported over $\calY$. For every $\alpha > 1$ and $B\ge 0$, let $\beta = \frac{\alpha}{\alpha - 1}$ be the H\"older conjugate of $\alpha$. The following statements are equivalent.
\begin{itemize}
    \item $D_{\alpha}(P\| Q) \le B$.
    \item For every function $h : \calY\to \mathbb{R}^{\ge 0}$, it holds that $\Ex_{y\sim P}[h(y)]\le e^{\frac{B(\alpha-1)}{\alpha}} \Ex_{y\sim Q}[h(y)^{\beta}]^{1/\beta}$.
\end{itemize}
\end{lemma}

Note that if we let $\alpha \to \infty$, then Lemma~\ref{lemma:renyi-charact} converges to a characterization of pure-DP. That is, $D_{\infty}(P\| Q)\le B$ if and only if $\Pr[P=y] \le e^B \Pr[Q=y]$ for every $y\in \calY$.

Lemma~\ref{lemma:renyi-charact} provides a convenient tool to reason about the privacy loss in an interactive environment consisting of multiple rounds. Intuitively, this is because Condition $2$ in the statement above is more amenable to a ``hybrid argument''. However, to quantify the privacy loss during an interaction, we still need to find a way to track the privacy loss.

\noindent\textbf{Measure theory setup.} Before we continue, it would be more convenient to switch to a measure-theoretic language. Consider two measures $P,Q$ on a space $\calY$ ($P$ and $Q$ are not necessarily probability measures), we say that $P$ is $\beta$-dominated by $Q$, denoted by $P\preceq_{\beta} Q$, if for every measurable function $f:\calY\to \mathbb{R}^{\ge 0}$, it holds that
\[
\| f\|_{P, 1} := \int f(y) dP(y) \le \left( \int f(y)^{\beta} dQ(y) \right)^{1/\beta} =: \| f\|_{Q,\beta}.
\]
When $\calY$ is a finite set, the integral coincides with an equivalent summation. i.e.,
\[
\int f(y) dP(y) = \sum_{y} P(y) f(y).
\]
We will use integral and summation interchangeably.

In this notation, Lemma~\ref{lemma:renyi-charact} can be equivalently stated as $D_{\alpha}(P\| Q) \le B$ if and only if $P$ is $\beta$-dominated by $e^BQ$ for $\beta = \frac{\alpha}{\alpha - 1}$  .

The following lemma is essential for us.

\begin{lemma}\label{lemma:renyi-tracker}
Let $\calY_1\times \calY_2$ be a space. Consider two distributions $P,Q$ on $\calY_1 \times \calY_2$. Assume $\supp(P) = \supp(Q) = \calY_1\times \calY_2$. Let $P_1,P_2$ be the margin of $P$ on $\calY_1,\calY_2$. For each $y_1\in \calY_1$, denote by $P_2|_{P_1=y_1}$ the marginal distribution of $y_2$ conditioning on $y_1$. Also define the same notation for $Q$. 

Let $\beta \ge 1, B\ge 0$ be two reals. Let $\alpha = \frac{\beta}{\beta - 1}$. For each $y_1\in \calY_1$, define 
\[
\ell_1(y_1) = \inf_K \left\{ K : P_2|_{P_1 = y_1} \preceq_{\beta} K \cdot  Q_2|_{Q_1 = y_1} \right\} = \exp({D_{\alpha}(P_2|_{P_1=y_1} \| Q_2|_{Q_1 = y_1})}).
\]
Suppose $P\preceq_\beta e^B Q$. Consider the measure spaces $(\calY_1,P_1(y_1)\cdot \ell_1(y_2)^{1/\beta})$ and $(\calY_1,Q_2)$. We have
\[
P_1\ell_1^{1/\beta} \preceq_{\beta} e^B Q_1.
\]
\end{lemma}

Intuitively, the function $\ell(y_1)$ serves as the role of ``privacy budget monitor''. To see this, fix an adversary $\calA$ and think of $(y_1,y_2)$ as the responses of the system to the adversary\footnote{Although the query made by $\calA$ is not explicitly recorded, the pair $(y_1,y_2)$ can capture this information by requiring that each response $y_i$ must be attached with the query message $x_i$. This does not leak any additional information because $x_i$ is solely chosen by $\calA$.}. After observing $y_1$, the adversary wants to distinguish between two conditional distributions $P_2|_{P_1 = y_1}$ and $Q_2|_{Q_1 = y_1}$. At this moment, $\ell(y_1)$ shows up as an upper bound of ``extra information'' that the adversary can extract by utilizing their second query. Alternatively, $\ell(y_1)$ quantifies the amount of the remaining privacy budget the mechanism has after outputting $y_1$. On average, the function $\ell_1(y_1)$ provides a fine-grained control of the privacy loss in the sense that $P_1\ell_1^{1/\beta} \preceq_{\beta} e^B Q_1$.

\noindent\textbf{Proof for a $3$-round toy example.} We are ready to describe the proof for Theorem~\ref{theo:renyi}. To illustrate the idea, we prove a toy case here and defer the full proof to Appendix~\ref{append:renyi}. The proof for the toy case includes all the important ideas. Extending it to a full proof is straightforward. 

We describe the toy scenario now. Suppose there are two mechanisms $\calM_1,\calM_2$ that run on a sensitive input bit $b\in \{0,1\}$. The interaction consists of $3$ rounds. The adversary $\calA$ communicates with $\calM_1,\calM_2,\calM_1$ in order, and outputs the response $(y_1,y_2,y_3)$. For brevity, we assume that each response $y_i$ contains a copy of the query message $x_i$, so that we recover the whole transcript $((x_1,y_1),(x_2,y_2),(x_3,y_3))$ only from the responses.

Let $P, Q\in \Delta(\calY\times \calY\times \calY)$ be the output distribution when $\calA$ interacts with $(\calM^0_1,\calM^0_2)$ and $(\calM^1_1,\calM^1_2)$, respectively. Suppose $\calM_1,\calM_2$ are $(\alpha,\eps_1)$, $(\alpha, \eps_2)$-R\'enyi DP respectively. Our goal is to prove that
\[
\max\left\{ D_{\alpha}(P\| Q), D_{\alpha}(Q \| P) \right\} \le \eps_1 + \eps_2.
\]
We bound $D_{\alpha}(P\| Q)$ below. The bound for $D_{\alpha}(Q\|P)$ is symmetric. Be Lemma~\ref{lemma:renyi-charact}, it suffices to show that for every $h : \calY\times \calY\to \calY \to \mathbb{R}^{\ge 0}$ that
\begin{align}
\sum_{y = (y_1,y_2,y_3)} P(y) h(y) \le \left( e^{\eps_1 + \eps_2}  \sum_{y = (y_1,y_2,y_3)} Q(y) h(y)^\beta \right)^{1/\beta} \label{eq:renyi-goal}
\end{align}
where $\beta = \frac{\alpha}{\alpha - 1}$ is the H\"older conjugate of $\alpha$. Let $P_1,P_2,P_3$ be the projection of $P$ onto the three rounds, and let $P_i|_{y_{<i}}$ denote the distribution of $y_i$ conditioning on $y_1,\dots, y_{i-1}$. Also define the same notation for $Q$. Then we write 
\begin{align}
\sum_{y = (y_1,y_2,y_3)} P(y) h(y) = \sum_{y_1} \left( P_1(y_1) \sum_{y_2} \left( P_2|_{y_1}(y_2) \sum_{y_3} P_3|_{y_{<3}}(y_3) h(y) \right) \right). \label{eq:renyi-rewrite}
\end{align}
For every $y_1\in \calY$, let $\calM^0_1|_{y_1}$ (resp. $\calM^1_1|_{y_1}$) denote the interactive system $\calM^0_1$ (resp. $\calM_1^1$) \emph{conditioning on} that it has answered $y_1$ to the first query (recall we have assumed that $y_1$ contains $x_1$). Formally, for every $b\in \{0,1\}$, $(x_2,y_2),\dots, (x_t,y_t)$ and $x_{t+1}$, define
\[
\calM^b_1|_{y_1}((x_j,y_j)_{2\le j\le t}, x_{t+1}) := \calM^b_1((x_j,y_j)_{1\le j\le t}, x_{t+1}).
\]
Next, define
\begin{align}
\ell_1(y_1) := \exp\left( \sup_{A: \text{adversary}} \left\{ D_{\alpha}\big( \IT(A:\calM^0_1|_{y_1}) \| \IT(A:\calM^1_1|_{y_1}) \big) \right\} \right). \label{eq:renyi-ell}
\end{align}
From Lemma~\ref{lemma:renyi-tracker}, one can show that $P_1 \ell_1^{1/\beta} \preceq e^B Q_1$. Turning back to \eqref{eq:renyi-goal}, we then have
\begin{align}
&~~~~  \sum_{y_1} \left( P_1(y_1) \sum_{y_2} \left( P_2|_{y_1}(y_2)  \sum_{y_3}P_3|_{y_{<3}}(y_3)  \underline{ h(y) } \right) \right) \label{eq:step_0} \\
&\le  \sum_{y_1} \left( P_1(y_1) \sum_{y_2} \left( P_2|_{y_1}(y_2)  ~\underline{ \left(\ell_1(y_1) \sum_{y_3} Q_3|_{y_{<3}}(y_3) h(y)^{\beta} ~\right)^{1/\beta}}  \right) \right) \label{eq:step_1}\\
&\le  \sum_{y_1} \left( P_1(y_1)  \left(e^{\eps_2}  \sum_{y_2} \left( Q_2|_{y_1}(y_2)  \ell_1(y_1)\sum_{y_3} Q_3|_{y_{<3}}(y_3) h(y)^{\beta}  \right) \right)^{1/\beta} \right)  \label{eq:step_2} \\
&=  \sum_{y_1} \left( P_1(y_1)  \ell_1(y_1)^{1/\beta} ~\underline{ \left(e^{\eps_2} \sum_{y_2} \left( Q_2|_{y_1}(y_2)  \sum_{y_3} Q_3|_{y_{<3}}(y_3) h(y)^{\beta}  \right)  \right)^{1/\beta} }~  \right)  \label{eq:step_25}\\
&\le  \left( e^{\eps_1+\eps_2} \sum_{y_1} \left( Q_1(y_1) \sum_{y_2}\left( Q_2|_{y_1}(y_2)  \sum_{y_3} Q_3|_{y_{<3}}(y_3) h(y)^{\beta}  \right) \right) \right)^{1/\beta} \label{eq:step_3} \\
&= \left( e^{\eps_1+\eps_2} \sum_{y} Q(y) h(y)^\beta \right)^{1/\beta}. \label{eq:final}
\end{align}
Here, we used inequalities of the form $\sum_{y} P(y)\cdot  h(y) \le \left(C\cdot  \sum_{y} Q(y) \cdot h(y)^{\beta}\right)^{1/\beta}$ three times (they are \eqref{eq:step_0} $\Rightarrow$ \eqref{eq:step_1} $\Rightarrow$ \eqref{eq:step_2} and \eqref{eq:step_25} $\Rightarrow$ \eqref{eq:step_3}). We use underlines to highlight the ``$h$'' part of each step in the deductions above. 

\eqref{eq:step_0} $\Rightarrow$ \eqref{eq:step_1} is the most critical step. To see this, observe that knowing $y_2$ does not change the view of the first mechanism, because the second query is sent to the independently running mechanism $\calM^b_2$. Therefore, $\calM^b_1|_{y_1}$ remains the same after conditioning on \emph{both} $y_1$ and $y_2$. Now, note that $P_3|_{y_{<3}}$ (resp. $Q_3|_{y_{<3}}$) exactly describes one round of interaction between the adversary and $\calM^0_1|_{y_1}$ (resp. $\calM^1_1|_{y_1}$). Consequently, the information leaked by $y_3$ must be subject to the bound \eqref{eq:renyi-ell} and the inequality holds. Having verified \eqref{eq:step_0} $\Rightarrow$ \eqref{eq:step_1}, the steps \eqref{eq:step_1} $\Rightarrow$ \eqref{eq:step_2} and \eqref{eq:step_25} $\Rightarrow$ \eqref{eq:step_3} are straightforward. 


Having justified \eqref{eq:final} for every measure function $h$, we conclude that $D_{\alpha}(P\| Q) \le e^{\eps_1 + \eps_2}$. A symmetric argument shows that $D_{\alpha}(Q\| P) \le e^{\eps_1 + \eps_2}$. This completes the proof for the toy example.

\noindent\textbf{Proof sketch for the general case.} The proof for the general case extends the idea above with some minor twists. By induction, we only need to prove the composition theorem for the case with two mechanisms and many rounds. An issue worth noting is that $\calA$ can choose the next query object based on previous responses. However, we can suppose without loss of generality that $\calA$ always communicates with mechanisms alternately, by adding a vanilla query $x^*$ to the query space. If the current mechanism is not the one $\calA$ wishes to speak with, $\calA$ just sends the vanilla query $x^*$. The mechanism then returns a fixed response, which does not leak any information. We refer to Appendix~\ref{append:renyi} for the detail of the proof.


\section{Conclusion and Future Directions}

In this work, we consider the concurrent composition of interactive mechanisms. Regarding the general privacy guarantee under the concurrent composition, our result gives optimal composition theorems for several popular definitions of differential privacy, including $(\eps,\delta)$-DP and R\'enyi DP. Our work is purely theoretical, and we do not see any negative societal impacts it may cause.

For future directions, we ask whether one can use our composition theorems to design new differentially-private algorithms that may involve running several differentially-private mechanisms in parallel. It is also interesting to explore more practical implications of the concurrent composition phenomena.


We also note that there is a recent interest in \emph{fully adaptive} compositions of differential privacy, which studies how the data analyst can manage the privacy budget and monitor the privacy loss themselves. In particular, the notion of privacy odometers and filters were proposed to capture these demands \citet{DBLP:conf/nips/RogersVRU16, DBLP:conf/nips/FeldmanZ21, WRRW-adaptive_compisition,Leuyer21-renyi_adaptive}. This question necessitates a better understanding of information leakage in an interactive environment. Our work developed several new tools and techniques to reason about interactive mechanisms. Can our technique be useful in studying fully adaptive compositions?


\section*{Acknowledgements}

I am grateful to Jelani Nelson for advising this project and providing useful comments on an early draft of this paper. I would also like to thank Salil Vadhan and Wanrong Zhang for insightful discussions about their work.

X. Lyu was supported by ONR DORECG award N00014-17-1-2127.

\bibliographystyle{plainnat}
\bibliography{neurips_2022}


\section*{Checklist}

\begin{enumerate}

\item For all authors...
\begin{enumerate}
  \item Do the main claims made in the abstract and introduction accurately reflect the paper's contributions and scope?
    \answerYes{}
  \item Did you describe the limitations of your work?
    \answerYes{We included a discussion for future directions.}
  \item Did you discuss any potential negative societal impacts of your work?
    \answerYes{}
  \item Have you read the ethics review guidelines and ensured that your paper conforms to them?
    \answerYes{}
\end{enumerate}

\item If you are including theoretical results...
\begin{enumerate}
  \item Did you state the full set of assumptions of all theoretical results?
    \answerYes{}
        \item Did you include complete proofs of all theoretical results?
    \answerYes{The full proofs will be available in supplementary material.}
\end{enumerate}

\item If you ran experiments...
\begin{enumerate}
  \item Did you include the code, data, and instructions needed to reproduce the main experimental results (either in the supplemental material or as a URL)?
    \answerNA{}
  \item Did you specify all the training details (e.g., data splits, hyperparameters, how they were chosen)?
    \answerNA{}
        \item Did you report error bars (e.g., with respect to the random seed after running experiments multiple times)?
    \answerNA{}
        \item Did you include the total amount of compute and the type of resources used (e.g., type of GPUs, internal cluster, or cloud provider)?
    \answerNA{}
\end{enumerate}

\item If you are using existing assets (e.g., code, data, models) or curating/releasing new assets...
\begin{enumerate}
  \item If your work uses existing assets, did you cite the creators?
    \answerNA{}
  \item Did you mention the license of the assets?
    \answerNA{}
  \item Did you include any new assets either in the supplemental material or as a URL?
    \answerNA{}
  \item Did you discuss whether and how consent was obtained from people whose data you're using/curating?
    \answerNA{}
  \item Did you discuss whether the data you are using/curating contains personally identifiable information or offensive content?
    \answerNA{}
\end{enumerate}

\item If you used crowdsourcing or conducted research with human subjects...
\begin{enumerate}
  \item Did you include the full text of instructions given to participants and screenshots, if applicable?
    \answerNA{}
  \item Did you describe any potential participant risks, with links to Institutional Review Board (IRB) approvals, if applicable?
    \answerNA{}
  \item Did you include the estimated hourly wage paid to participants and the total amount spent on participant compensation?
    \answerNA{}
\end{enumerate}

\end{enumerate}



\appendix

\section{Appendix: Missing Proofs}\label{append:proof}

In this appendix, we show the formal proofs for all the lemmas and claims in the main paper.

\subsection{Proofs for Approximate Differential Privacy}\label{append:approxiate}

This subsection present omitted proofs in Section~\ref{sec:proof-approximate}.

\subsubsection{The key lemma}\label{sec:key-lemma}

We start with the proof for Lemma~\ref{lemma:decomposition}. 
 
\begin{reminder}{Lemma~\ref{lemma:decomposition}}
Suppose $\calM^0,\calM^1$ are $(\eps,\delta)$-indistinguishable. There are two systems $\calE^0,\calE^1$ satisfying the following.
\begin{itemize}
    \item For every adversary $\calA$ and $b\in \{0,1\}$, it holds that $\IT(\calA : \calM^b) \ge \delta\cdot  \IT(\calA : \calE^b)$.
    \item For every adversary $\calA$, every set of transcripts $S\subseteq \{(x_i,y_i)_{i\in [T]}\}$ and $b\in \{0,1\}$, it holds that
    \[
    \begin{aligned}
    &  ~~~~ \Pr[\IT(\calA : \calM^b)\in S] - \delta \Pr[\IT(\calA:\calE^b)\in S]  \\
    & \le e^\eps \left( \Pr[\IT(\calA : \calM^{1-b} ) \in S] - \delta \Pr[\IT(\calA : \calE^{1-b})\in S] \right).
    \end{aligned}
    \]
\end{itemize}
\end{reminder}
\begin{proof} Without loss of generality, we assume that the interaction between $\calM^{0/1}$ and $\calA$ consists of exactly $T\in \mathbb{N}$ rounds. For every $t\in [T]$, each $(x_i,y_i)_{i\in [t]}$ and $b\in \{0,1\}$, denote 
\begin{align}
M^b((y_i)_{i\in [t]},(x_i)_{i\in [t]}) := \prod_{i=1}^{t} \Pr[\calM^{b}((x_j,y_j)_{j < i}, x_i) = y_i]. \label{eq:def-Mb}
\end{align}
Intuitively, $M^b((y_i)_{i\in [t]},(x_i)_{i\in [t]})$ is the probability of $\calM^b$ responding $(y_1,\dots, y_t)$, conditioning on that the query messages are fixed to $(x_1,\dots, x_t)$. Note that knowing $M^b((y_i)_{i\in [T]},(x_i)_{i\in [T]})$ for every $(x_i,y_i)_{i\in [T]}$ \emph{uniquely} determines the system.

Let $\calA$ be an arbitrary adversary. For each $(x_i,y_i)_{i\in [t-1]}$ and $x_t$, denote
\begin{align}
A((x_i)_{i\in [t]},(y_i)_{i\in [t-1]}) := \prod_{i=1}^t \Pr[\calA((x_j,y_j)_{j < i}) = x_i]. \label{eq:def-Ab}
\end{align}
Note that $A((x_i)_{i\in[t]},(y_i)_{i\in [t-1]})$ is the probability of $\calA$ sending queries $(x_1,\dots, x_t)$, conditioning on that the responses to the first $t-1$ queries are fixed to $(y_1,\dots, y_{t-1})$. 

In the following, when the size of a list $(\ell_i)_{i\in [L]}$ is clear from context, we may omit the subscript and simply write $(\ell_i)$ to denote the list. Now, having defined \eqref{eq:def-Mb} and \eqref{eq:def-Ab}, we observe for each transcript $(x_i,y_i)_{i\in [T]}$ that
\begin{align}
\Pr[\IT(\calA,\calM^b) = (x_i,y_i)_{i\in [T]}] = M^b((y_i),(x_i))\cdot  A((x_i),(y_i)). \label{eq:decouple}
\end{align}

In the following, we will construct two systems $\calE^{0/1}$ such that, for every $(x_i,y_i)_{i\in [T]}$, it holds that
\begin{align}
M^b((y_i),(x_i)) \ge \delta E^b((y_i),(x_i))\label{eq:goal-1}
\end{align}
and 
\begin{align}
M^b((y_i),(x_i))-\delta E^b((y_i),(x_i)) \le e^\eps \big(M^{1-b}((y_i),(x_i)) - \delta E^{1-b}((y_i),(x_i))\big). \label{eq:goal-2}
\end{align}
If we have two systems $\calE^{0/1}$ satisfying the above, then we can verify that they satisfy the lemma statement by combining \eqref{eq:decouple}, \eqref{eq:goal-1} and \eqref{eq:goal-2}.

Now we describe the construction. We start by defining for each $b\in \{0,1\}$, $t\le T$ and every partial history $(x_i,y_i)_{i\le t-1} \in (\calX\times \calY)^{t-1}$, $x_{t}\in \calX$ a control function as
\begin{align}
    \Lower^b((x_i,y_i)_{i< t}, x_{t}) := 
    \begin{cases}
    \sum_{y_{t}\in \calY} \max_{x_{t+1}} \{\Lower^b((x_i,y_i)_{i\le t}, x_{t+1}) \} & \text{$t < T$} \\
    \sum_{y_t\in \calY} \max\left\{ M^b((y_i),(x_i)) - e^\eps M^{1-b}((y_i),(x_i)),0\right\} & \text{$t = T$} 
    \end{cases}. \label{eq:def-lower}
\end{align}
For every $t\le T-1$ and $(x_i,y_i)_{i\le t}$, we also define the following control function:
\begin{align}
    & \Upper^b((x_i,y_i)_{i\le t}) := 
     M^b((y_i)_{i\le t},(x_i)_{i\le t})  - e^{-\eps} M^{1-b}((y_i)_{i\le t},(x_i)_{i\le t}).
\end{align}

We need the following two facts regarding the control functions.

\begin{claim}\label{claim:start-good}
For each $b\in \{0,1\}$ and $x_1\in \calX$, it holds that $\Lower^b(\emptyset, x_1)\le \delta$.
\end{claim}

\begin{proof}
We construct an adversary $\calA^*$ as follows. $\calA^*$ is deterministic. It always sends $x_1$ as the first query. For every $1\le t\le T-1$ and history $(x_i,y_i)_{i\in [t]}$, $\calA^*$ computes the next query as
\[
\calA^*((x_i,y_i)_{i\in [t]}) = \arg\max_{x_{t+1}}\{ \Lower^b((x_i,y_i)_{i\le t}, x_{t+1}) \}.
\]
Now, define 
\[
S^b := \{ (x_i,y_i)_{i\in [T]} : \Pr[\IT(\calA^*,\calM^b)=  (x_i,y_i)_{i\in [T]} > e^\eps \Pr[\IT(\calA^*,\calM^{1-b}) = (x_i,y_i)_{i\in [T]}] \}.
\]
Given that $\calM^0$ and $\calM^1$ are $(\eps,\delta)$-indistinguishable, we know that
\[
\sum_{(x_i,y_i)\in S^b} \Pr[\IT(\calA^*,\calM^b)=  (x_i,y_i)_{i\in [T]}] - e^\eps \Pr[\IT(\calA^*,\calM^{1-b}) = (x_i,y_i)_{i\in [T]}] \le \delta.
\]
On the other hand, by the definition of $\calA^*$ and \eqref{eq:def-lower}, it holds that
\[
\Lower^b(\emptyset, x_1) = \sum_{(x_i,y_i)\in S^b} \Pr[\IT(\calA^*,\calM^b)=  (x_i,y_i)] - e^\eps \Pr[\IT(\calA^*,\calM^{1-b}) = (x_i,y_i)].
\]
This can be verified by tracing how $\Lower^b(\emptyset, x_1)$ is determined from queries $(x_1,\dots, x_T)$ (in the ``max'' operator), and noting that $\calA^*$ follows exactly the same queries. Combining two equations above concludes the proof of Claim.
\end{proof}

\begin{claim}\label{claim:up-good}
For every $b\in \{0,1\}$, $t\le T-1$ and $(x_i,y_i)_{i\le t-1} \in (\calX\times \calY)^{t-1}, x_t\in \calX$, it holds that
\[
\Lower^b((x_i,y_i)_{i<t},x_{t}) \le \Upper^b((x_i,y_i)_{i<t}) + e^{-\eps} \Lower^{1-b}((x_i,y_i)_{i<t},x_t).
\]
\end{claim}

\begin{proof}
We prove this claim by \emph{downward} induction on $t$. For the case $t=T-1$, we have by definition that
\begin{align}
&~~~~ e^{-\eps} \Lower^{1-b}((x_i,y_i)_{i<t},x_{t}) \notag \\
& = \sum_{y_t\in\calY} \max\left\{ e^{-\eps} M^{1-b}((y_i),(x_i)) - M^{b}((y_i),(x_i)), 0 \right\} \label{eq:ded-1-def}\\
& = \sum_{y_t\in \calY} e^{-\eps} M^{1-b}((y_i),(x_i)) - M^{b}((y_i),(x_i))  + \notag \\
&~~~~ \sum_{y_t\in \calY} \max\left\{ M^b((y_i),(x_i)) - e^{-\eps}M^{1-b}((y_i),(x_i)), 0 \right\} \label{eq:ded-2-trick} \\
&\ge - \Upper^b((x_i,y_i)_{i<t}) + \Lower^{b}((x_i,y_i)_{i<t},x_t). \label{eq:ded-3-def}
\end{align}

We justify the deductions briefly. \eqref{eq:ded-1-def} is by definition. \eqref{eq:ded-2-trick} uses a simple trick that $\max\left\{ a,0 \right\} = a + \max\left\{ -a,0 \right\}$. The last step \eqref{eq:ded-3-def} is by definition again. In particular, we observe that for every $x_T\in \calX$, it holds that
\[
\sum_{y_T\in \calY} M^b((y_i),(x_i)) = \prod_{i=1}^{T-1} \Pr[\calM^{b}((x_j,y_j)_{j < i}, x_i) = y_i] .
\]
This proves the base case for $t = T-1$.

Assume the claim is true for $t + 1\le T-1$. We consider the case of $t$. We have
\[
\begin{aligned}
\Lower^b((x_i,y_i)_{i<t},x_{t}) 
&= \sum_{y_t} \max_{x_{t+1}} \{ \Lower^b((x_i,y_i)_{i\le t}, x_{t+1}) \} \\
&\le \sum_{y_t} \max_{x_{t+1}} \{ \Upper^b((x_i,y_i)_{i\le t}) + e^{-\eps} \Lower^{1-b}((x_i,y_i)_{i\le t},x_{t+1}) \} \\
&\le \sum_{y_t} \Upper^b((x_i,y_i)_{i\le t}) + e^{-\eps}  \max_{x_{t+1}}\{\Lower^{1-b}((x_i,y_i)_{i\le t},x_{t+1})\} \\
&= \Upper^b((x_i,y_i)_{i<t}) + e^{-\eps} \Lower^{1-b}((x_i,y_i)_{i<t},x_t).
\end{aligned}
\]
The first inequality is due to the induction hypothesis. The second inequality is straightforward. This completes the proof for the claim.
\end{proof}

\medskip\noindent\textbf{The construction.} We are ready to describe the construction. In the following, we will assume $\eps > 0$. Having shown the construction for every $\eps > 0$, the case for $\eps = 0$ can be argued by continuity. We will construct $\calE^0,\calE^1$ by specifying for every $t\in [T]$ and $(x_i,y_i)_{i\le t}$ the following:
\[
E^b((y_i)_{i\le t},(x_i)_{i\le t}) := \prod_{i=1}^{t} \Pr[\calE^b((x_j,y_j)_{j<i}, x_i) = y_i].
\]
Note that a valid $E^b(\cdot)$ uniquely defines a system $\calE^b$. For brevity, we also define $E^0(\emptyset) = E^1(\emptyset) = 1$. Intuitively, we use $\emptyset$ to denote two ``empty lists'' (i.e., two lists $(y_i)_{i\le t}, (x_i)_{i\le t}$ with $t = 0$).

We will construct $E^b((y_i)_{i\le t},(x_i)_{i\le t})$ for $t=1,2,\dots, T$ in order. Throughput the construction, we maintain the following property. For every $0\le t\le T$, $(x_i,y_i)_{i\le t}$ and $b\in \{0,1\}$, we require
\begin{align}
    \delta \cdot  E^b((y_i)_{i\le t},(x_i)_{i\le t}) \ge \begin{cases}
     \max_{x_{t+1}\in \calX}\{\Lower^b((x_j,y_j)_{j\le t}, x_{t+1})\} & t < T \\
     \max \left\{ M^b((y_i),(x_i)) - e^\eps M^{1-b}((y_i),(x_i)), 0 \right\} & t = T
    \end{cases} \label{eq:require-1}
\end{align}
and
\begin{align}
\delta\cdot  E^b((y_i)_{i\le t},(x_i)_{i\le t}) \le \Upper^b((x_i,y_i)_{i\le t}) + e^{-\eps} \delta \cdot  E^{1-b}((y_i)_{i\le t}, (x_i)_{i\le t}). \label{eq:require-2}
\end{align}

Meanwhile, for $E^b((y_i),(x_i))$ to describe a valid system, it is necessary and sufficient for it to be non-negative and satisfy the following equation for every $(x_i,y_i)_{i\le t}\in (\calX\times \calY)^{t}$ and $x_{t+1}$:
\begin{align}
\sum_{y_{t+1}\in \calY}E^b((y_i)_{i\le t+1},(x_i)_{i\le t+1}) = E^b((y_i)_{i\le t}, (x_i)_{i\le t}). \label{eq:require-3}
\end{align}

Next, we shall prove that we can construct a valid $E^{0/1}$ satisfying \eqref{eq:require-1}, \eqref{eq:require-2} and \eqref{eq:require-3}. As we have said, we will construct $E^b$ gradually in the increasing order of $t\in [T]$. For $t = 0$, we have set $E^0(\emptyset) = E^1(\emptyset) = 1$. \eqref{eq:require-1} holds by Claim~\ref{claim:start-good}, and \eqref{eq:require-2} holds trivially.

\newcommand{\tE}{\widetilde{E}}

Now let $t<T$. Also let $(y_i)_{i\le t}\in \calY^t, (x_i)_{i\le t}\in \calX^t$ be two lists. Suppose we have constructed $E^{0/1}((y_i)_{i\le t},(x_i)_{i\le t})$ that satisfies \eqref{eq:require-1} and \eqref{eq:require-2}. For every $x_{t+1}\in \calX$ and $y_{t+1}\in \calY$, we construct $E^{0/1}((y_i)_{i\le t+1},(x_i)_{i\le t+1})$ in the following.

Fix $x_{t+1}\in \calX$. We temporarily set
\[
\tE^b((y_i)_{i\le t+1},(x_i)_{i\le t+1}) = \frac{1}{\delta}\cdot \begin{cases}
     \max_{x_{t+2}\in \calX}\{\Lower^b((x_j,y_j)_{j\le t+1}, x_{t+2})\} & t+1 < T \\
     \max \left\{ M^b((y_i),(x_i)) - e^\eps M^{1-b}((y_i),(x_i)), 0 \right\} & t+1 = T
    \end{cases}. 
\]
By Claim~\ref{claim:up-good}, we know that $\tE^b$ satisfies \eqref{eq:require-2}. By the construction, $\tE^b$ satisfies \eqref{eq:require-1}. However, $\tE^b$ may fail to satisfy \eqref{eq:require-3}. Still, we have
\[
\sum_{y_{t+1}} \tE^{b}((y_i)_{i\le t+1},(x_i)_{i\le t+1}) \le \frac{1}{\delta} \Lower^b((x_i,y_i)_{i\le t}, x_{t+1}) \le E^b((y_i)_{i\le t},(x_i)_{i\le t}).
\]
In the following, we show that one can adjust $\tE^{b}$ by increasing some $\tE^{b}((y_i)_{i\le t+1},(x_i)_{i\le t+1})$ properly, so that the new $\tE^b$ satisfies all of \eqref{eq:require-1}, \eqref{eq:require-2} and \eqref{eq:require-3}. 

To begin with, we define for each $b\in \{0,1\}$ the quantity
\begin{align}
    \mathrm{Gap}_b :=  E^b((y_i)_{i\le t},(x_i)_{i\le t}) - \sum_{y_{t+1}} \tE^{b}((y_i)_{i\le t+1},(x_i)_{i\le t+1}).
\end{align}
Our goal is to decrease $\mathrm{Gap}_0,\mathrm{Gap}_1$ to zero by increasing $\tE$. Since we only increase $\tE^b((y_i)_{i\le t+1},(x_i)_{i\le t+1})$, \eqref{eq:require-1} can never be compromised and we only need to consider \eqref{eq:require-2}. Consider $y_{t+1}$ and $b\in \{0, 1\}$. We say that $\tE^{b}$ is \emph{tight} at $y_{t+1}$, if 
\begin{align}
\delta \tE^b((y_i)_{i\le t+1},(x_i)_{i\le t+1}) = \Upper^b((x_i,y_i)_{i\le t+1}) + e^{-\eps} \delta \tE^{1-b}((y_i)_{i\le t+1}, (x_i)_{i\le t+1}). \notag
\end{align}
Intuitively, $\tE^b$ being tight at $y_{t+1}$ means that we cannot increase $\tE^b((y_i)_{i\le t+1},(x_i)_{i\le t+1})$ without increasing $\tE^{1-b}((y_i)_{i\le t+1}, (x_i)_{i\le t+1})$.

Here shows our adjustment strategy. We consider each $y_{t+1}\in \calY$ in an arbitrary but fixed order. For each $y_{t+1}$, we gradually increase $ \tE^{0/1}((y_i)_{i\le t+1},(x_i)_{i\le t+1})$ until one of the following events happens.
\begin{itemize}
    \item Both $\tE^{0}$ and $\tE^1$ get tight at $y_{t+1}$.
    \item $\mathrm{Gap}_0 = 0$, and $\tE^1$ is tight at $y_{t+1}$.
    \item $\mathrm{Gap}_1 = 0$, and $\tE^0$ is tight at $y_{t+1}$.
    \item $\mathrm{Gap}_0 = \mathrm{Gap}_1 = 0$.
\end{itemize}
It is easy to see that if none of the above happens, we can keep increasing $\tE^{0/1}$ at $y_{t+1}$\footnote{To see this, note that for $b\in \{0,1\}$, we can keep increasing $\widetilde{E}^b$ until either (1) $\widetilde{E}^b$ gets tight at $y_{t+1}$, or (2) $\mathrm{Gap}_b = 0$. Therefore, if we cannot increase both $\widetilde{E}^0$ and $\widetilde{E}^1$, it must be one of the four cases above.}. This completes the description of the adjustment strategy.

Now, we claim that after the adjustment, we must have $\mathrm{Gap}_0 = \mathrm{Gap}_1 = 0$. Suppose it is not the case. For example, suppose $\mathrm{Gap}_0 \ne 0$. Then we know that $\tE^0$ is tight at every $y_{t+1}$. This means that
\[
\sum_{y_{t+1}} \delta \tE^{0}((y_i)_{i\le t+1},(x_i)_{i\le t+1}) = \Upper^0((x_i,y_i)_{i\le t}) + e^{-\eps} \sum_{y_{t+1}} \delta \tE^{1}((y_i)_{i\le t+1},(x_i)_{i\le t+1}).
\]
Recall that
\[
\delta E^0((y_i)_{i\le t},(x_i)_{i\le t}) \le \Upper^0((x_i,y_i)_{i\le t}) + e^{-\eps} \delta E^{1}((y_i)_{i\le t}, (x_i)_{i\le t}).
\]
Subtracting the inequality with the equality above, we deduce that $\mathrm{Gap}_0 \le e^{-\eps} \mathrm{Gap}_1$. It implies that $\mathrm{Gap}_1 > 0$. And we can use a symmetric argument to show that $\mathrm{Gap}_1 \le e^{-\eps}\mathrm{Gap}_0$. Since $e^{-\eps} < 1$, the only solution to the system of inequalities is $\mathrm{Gap}_0 = \mathrm{Gap}_1 = 0$, a contradiction!

Having proven the claim, we know there is a way to adjust $\tE^{0/1}$ so that they satisfy all of \eqref{eq:require-1}, \eqref{eq:require-2}, \eqref{eq:require-3}. We then set $E^{0/1}((y_i)_{i\le t+1},(x_i)_{i\le t+1})$ to be $\tE^{0/1}((y_i)_{i\le t+1},(x_i)_{i\le t+1})$ and finish the construction for $(x_i,y_i)_{i\le t}$ and $x_{t+1}$.

We use the construction above for $t =0,1,\dots, T-1$ in order to construct $E^{0/1}$. It remains to verify that $E^{0/1}$ satisfies the lemma statement. It suffices to verify for every $(x_i,y_i)_{i\le T} \in (\calX\times \calY)^T$ and $b\in \{0,1\}$ that
\[
M^{b}((y_i),(x_i)) \ge \delta E^b((y_i),(x_i))
\]
and
\[
(M^{b}((y_i),(x_i)) - \delta E^b((y_i),(x_i))) \le e^{\eps} (M^{1-b}((y_i),(x_i)) - \delta E^{1-b}((y_i),(x_i))).
\]
In fact, since $e^{\eps} > 1$, it suffices to verify the second inequality for $b\in \{0,1\}$. This can be verified by utilizing \eqref{eq:require-2}: note that $\Upper^b((x_i,y_i)_{i\le T}) = M^b((y_i),(x_i)) - e^{-\eps} M^{1-b}((y_i),(x_i))$, and \eqref{eq:require-2} tells us that
\[
\delta E^b((y_i),(x_i)) \le M^b((y_i),(x_i)) - e^{-\eps} M^{1-b}((y_i),(x_i)) + e^{-\eps} \delta E^{1-b}((y_i),(x_i)).
\]
Re-arranging proves the desired inequality.
\end{proof}

\begin{remark}
Note that to verify the correctness of $E^{0/1}$, we only used the condition \eqref{eq:require-2}. It seems that \eqref{eq:require-1} is useless in this proof. However, note that it is possible that $\Upper^b((y_i),(x_i))$ is negative for some $(x_i,y_i)_{i\le t}$, which makes it unclear whether \eqref{eq:require-2} can always be satisfied by a positive valuation of $E$. This is why we need the other control function $\Lower$.
\end{remark}

\subsubsection{Wrap-up}

Next, we quickly prove Lemma~\ref{lemma:system-subtraction}.

\begin{reminder}{Lemma~\ref{lemma:system-subtraction}}
Suppose $\calM, \calE$ are two systems such that for every adversary $\calA$, it holds that $\IT(\calA :\calM) \ge \delta \IT(\calA : \calE)$. Then there is a system $\calN$ such that for every adversary $\calA$, it holds that
\[
\begin{aligned}
\IT(\calA : \calM) \equiv \delta \IT(\calA : \calE)  + (1-\delta) \IT(\calA : \calN).
\end{aligned}
\]
\end{reminder}

\begin{proof}
We follow the notation in Section~\ref{sec:key-lemma}. Namely, for each $(x_i,y_i)_{i\le t}$, define 
\[
M((y_i)_{i\le t},(x_i)_{i\le t}) = \prod_{i=1}^t \Pr[\calM((x_j,y_j)_{j<i},x_i) = y_i].
\]
Also define the same notation for $E$. Then we construct
\[
N((y_i)_{i\le t},(x_i)_{i\le t}) = \frac{1}{1-\delta} \left( M((y_i)_{i\le t},(x_i)_{i\le t}) - \delta E((y_i)_{i\le t},(x_i)_{i\le t}) \right).
\]
Since $\calM \ge \delta \calE$, we know that $N((y_i),(x_i))$ is always non-negative. Moreover, $N$ encodes a valid system because
\[
\begin{aligned}
& ~~~~ \sum_{y_{t+1}\in \calY}N((y_i)_{i\le t+1},(x_i)_{i\le t+1}) \\
&=  \frac{1}{1-\delta} \sum_{y_{t+1}\in \calY}  M((y_i)_{i\le t+1},(x_i)_{i\le t+1}) - \delta E((y_i)_{i\le t+1},(x_i)_{i\le t+1})  \\
&= \frac{1}{1-\delta} \left( M((y_i)_{i\le t},(x_i)_{i\le t}) - \delta E((y_i)_{i\le t},(x_i)_{i\le t}) \right) \\
&= N((y_i)_{i\le t},(x_i)_{i\le t}).
\end{aligned}
\]
Finally, it is easy to verify $\IT(\calA : \calM) \equiv \delta \IT(\calA : \calE) + (1-\delta) \IT(\calA:  \calM)$.
\end{proof}

As we have shown in Section~\ref{sec:proof-approximate}, combining Lemma~\ref{lemma:decomposition}, \ref{lemma:system-subtraction} and \ref{lemma:pure-decompose} together, we can prove Theorem~\ref{theo:decomposition} easily. Next, we show how Theorem~\ref{theo:decomposition} implies Theorem~\ref{theo:approximate}.

\begin{proofof}{Theorem~\ref{theo:approximate}}
Let $\calM_1,\dots, \calM_k$ be $k$ mechanisms, where $\calM_i$ is $(\eps_i,\delta_i)$-approximate differentially private. We assume without loss of generality that all of $\calM_i$'s hold a bit $b\in \{0,1\}$ as the sensitive data.

Let $\calA$ be an arbitrary adversary interacting with $\COMP(\calM_1,\dots, \calM_k)$. Next, we show how one can simulate $\IT(\calA, \COMP(\calM_1^b,\dots, \calM_k^b))$ by running $k$ (approximate) randomized response mechanisms. For each $\calM_i^b$, construct an approximate randomized response mechanism $\RR^b_{\eps_i,\delta_i}$. The output distribution of $\RR^b_{\eps_i,\delta_i}$ is:
\[
\RR^b_{\eps_i,\delta_i} = \begin{cases}
(b, \top) & \text{w.p. $\delta$} \\
(b, \perp) & \text{w.p. $(1-\delta) \frac{e^\eps}{1+e^\eps}$} \\
(1-b, \perp) & \text{w.p. $(1-\delta) \frac{1}{1+e^{\eps}}$}
\end{cases}.
\]
We also prepare the decomposition of $\calM_i^{0/1}$ with $\calN_i^{0/1},\calE_i^{0/1}$ as promised by Theorem~\ref{theo:decomposition}.

Now, we construct a simulator $\calS$ as follows. For each $i\in [k]$, $\calS$ runs $\RR^b_{\eps_i,\delta_i}$ and gets a pair $(b_i, \sigma_i)$. If $\sigma_i = \top$, then let $\calB_i \gets \calE_i^{b_i}$. Otherwise, let $\calB_i \gets \calN_i^{b_i}$. In this way, $\calS$ gets a list of $k$ systems $(\calB_1,\dots, \calB_k)$. The simulator then simulates the interaction between $\calA$ and $\calB_1,\dots, \calB$, and outputs the interaction history. Let $\Output(\calS, b)$ denote the output distribution of $\calS$. We claim that
\begin{align}
\Output(\calS, b) \equiv \IT(\calA : \calM^b_1,\dots, \calM^b_k). \label{eq:to-prove}
\end{align}
To see this, for each $\calM^b_i$, consider a two-party communication, where one party is $\calM^b_i$, and the other party consists of $\calA$ and $\calM^b_{j}$ for $j\ne i$. The second party simulates all the interactions between $\calA$ and $\calM^{b}_{j}$, and only sends queries to $\calM^b_i$ when $\calA$ queries $\calM^b_i$. From the second party's viewpoint, $\calM^b_i$ looks identical to $ \delta_i\calE_i^b + (1-\delta_i)\frac{e^\eps}{1+e^\eps}\calN_i^b + (1-\delta_i)\frac{1}{1+e^\eps} \calN_i^{1-b}$. Therefore,
\[
\IT(\calA : \calM^b_1,\dots, \calM^b_i,\dots,\calM^b_k) \equiv \sum_{j=1}^3 p_j \cdot \IT(\calA : \calM^b_1,\dots, \calM^b_{i,j},\dots,\calM^b_k).
\]
Here, we use $(p_1,p_2,p_3) = (\delta_i,(1-\delta_i)\frac{e^\eps}{1+e^\eps}, (1-\delta_i) \frac{1}{1+e^\eps})$ and $(\calM^b_{i,1},\calM^b_{i,2},\calM^b_{i,3}) =(\calE_i^b, \calN_i^b,\calN_i^{1-b})$ for convenience. Applying this decomposition for every $i\in [k]$ proves \eqref{eq:to-prove}.

Finally, note that $\Output(\calS, b)$ is just a post-processing of the sequential composition of $k$ (approximate) randomized response mechanisms. Hence, the optimal sequential composition theorem holds for $\Output(\calS,b)$, which completes the proof.
\end{proofof}

\subsection{Proofs for R\'enyi Differential Privacy}\label{append:renyi}

In this section, we show omitted proofs for Theorem~\ref{theo:renyi}. 

\subsubsection{Preliminaries}

We need some technical preparations first. Consider a measure space $(X,\mu)$. For two measurable functions $f,g$, define their inner product as
\[
\langle f, g\rangle_{\mu} = \int f\cdot g d\mu.
\]
For a real $\alpha\ge 1$, define the $\ell_{\alpha}$-norm of a function $f$ as
\[
\| f\|_{\mu,\alpha} := \left( \int f^{\alpha} d\mu\right)^{1/\alpha}.
\]
Recall H\"older's inequality, which is essential for our proof.
\begin{fact}
Suppose $\alpha,\beta \ge 1$ are H\'older conjugates of each other (i.e., $\frac{1}{\alpha} + \frac{1}{\beta} = 1$). Suppose $f, g$ are two measurable functions. Then we have
\[
\langle f, g\rangle_{\mu} \le \| f\|_{\mu, \alpha} \cdot \| g\|_{\mu,\beta}.
\]
The inequality is sharp in the sense that for every measurable function $f$, we have
\[
\| f \|_{\mu,\alpha} = \sup_{h:h\not\equiv 0} \frac{\langle f, h\rangle_{\mu}}{\|h\|_{\mu,\beta}}.
\]
\end{fact}

Recall our definition of dominance. For two measures $P,Q$ on a space $\calY$, we say that $P$ is $\beta$-dominated by $Q$, denoted by $P\preceq_{\beta} Q$, if for every measurable function $f:\calY\to \mathbb{R}^{\ge 0}$, it holds that $\| f\|_{P, 1} \le \| f\|_{Q,\beta}$.

\subsubsection{Proof for lemmas}

We are ready to show the proofs now. We start with Lemma~\ref{lemma:renyi-charact}.

\begin{reminder}{Lemma~\ref{lemma:renyi-charact}}
Suppose $P, Q$ are two distributions supported on $\calY$. For every $\alpha > 1$ and $B\ge 0$, let $\beta = \frac{\alpha}{\alpha - 1}$ be the H\"older conjugate of $\alpha$. The following statements are equivalent.
\begin{itemize}
    \item $D_{\alpha}(P\| Q) \le B$.
    \item For every function $h : \calY\to \mathbb{R}^{\ge 0}$, it holds that $\Ex_{y\sim P}[h(y)]\le e^{\frac{B(\alpha-1)}{\alpha}} \Ex_{y\sim Q}[h(y)^{\beta}]^{1/\beta}$.
\end{itemize}
\end{reminder}

\begin{proof}
First, if there is $y\in \calY$ such that $0 = \Pr[Q=y] < \Pr[P = y]$, then we have $D_{\alpha}(P\| Q) = \infty$ and Condition 2 does not hold for any $B < \infty$. In the following, we assume $\supp(P) = \supp(Q) = \calY$. Note that in this case, we have $D_{\alpha}(P\| Q) < \infty$.

We write $P(y), Q(y)$ as shorthands for $\Pr[P=y]$ and $\Pr[Q=y]$ for brevity. Now, note that $D_{\alpha}(P\|Q) \le B$ is equivalent to $e^{D_{\alpha}(P\| Q)}\le e^B$, which is further equivalent to
\[
\Ex_{y\sim Q}\left[\frac{P(y)^{\alpha}}{Q(y)^{\alpha}}\right]^{1/\alpha} = \Ex_{y\sim P}\left[\frac{P(y)^{\alpha-1}}{Q(y)^{\alpha-1}}\right]^{1/\alpha} \le e^{\frac{B(\alpha-1)}{\alpha}}.
\]
Consider the measure space $M=(\calY, Q)$. By Holder's inequality, we have
\[
\Ex_{y\sim Q}\left[\left(\frac{P(y)}{Q(y)}\right)^{\alpha}\right]^{1/\alpha} = \left\| \frac{P}{Q} \right\|_{Q,\alpha} = \sup_{h:h\not\equiv 0} \left\{ \frac{\langle h, \frac{P}{Q}\rangle_Q }{\| h \|_{Q,\beta}} \right\}.
\]
Moreover, since $\frac{P(y)}{Q(y)}$ is non-negative, it suffices to consider only non-negative $h$ in the supremum above. Now we are ready to verify the equivalence.
\begin{itemize}
    \item If Condition $1$ holds, we have
    \[
    \sup_{h:h\not\equiv 0} \left\{ \frac{\langle h, \frac{P}{Q}\rangle_Q }{\| h \|_{Q,\beta}} \right\} =  \left\| \frac{P}{Q} \right\|_{Q,\alpha} \le e^{\frac{B(\alpha-1)}{\alpha}}.
    \]
    Therefore, for every $h:\calY\to \mathbb{R}^{\ge 0}$, it holds that 
    \[
    \Ex_{y\sim P}[h(y)] = \Ex_{y\sim Q}\left[ h(y)\cdot \frac{P(y)}{Q(y)} \right] \le \left\| \frac{P}{Q} \right\|_{Q,\alpha}  \| h \|_{Q,\beta} \le e^{\frac{B(\alpha-1)}{\alpha}} \Ex_{y\sim Q}[h(y)^{\beta}]^{1/\beta}.
    \]
    \item On the other hand, if Condition $2$ holds, we have
    \[
    \left\| \frac{P}{Q} \right\|_{Q,\alpha} = \sup_{h:h\not\equiv 0} \left\{ \frac{\langle h, \frac{P}{Q}\rangle_Q }{\| h \|_{Q,\beta}} \right\} \le e^{\frac{B(\alpha-1)}{\alpha}}.
    \]
\end{itemize}
This completes the proof.
\end{proof}

The next lemma is Lemma~\ref{lemma:renyi-tracker}.

\begin{reminder}{Lemma~\ref{lemma:renyi-tracker}}
Let $\calY_1\times \calY_2$ be a space. Consider two distributions $P,Q$ on $\calY_1 \times \calY_2$. Assume $\supp(P) = \supp(Q) = \calY_1\times \calY_2$. Let $P_1,P_2$ be the margin of $P$ on $\calY_1,\calY_2$. For each $y_1\in \calY_1$, denote by $P_2|_{P_1=y_1}$ the marginal distribution of $y_2$ conditioning on $y_1$. Also define the same notation for $Q$. 

Let $\beta \ge 1, B\ge 0$ be two reals. For each $y_1\in \calY_1$, define 
\[
\ell_1(y_1) = \inf_K \left\{ K : P_2|_{P_1 = y_1} \preceq_{\beta} K \cdot  Q_2|_{Q_1 = y_1} \right\}.
\]
Suppose $P\preceq_\beta e^B Q$. Consider the measure spaces $(\calY_1,P_1(y_1)\cdot \ell_1(y_2)^{1/\beta})$ and $(\calY_1,Q_2)$. We have
\[
P_1\ell_1^{1/\beta} \preceq_{\beta} e^B Q_1.
\]
\end{reminder}
\begin{proof}
Suppose by contradiction that the conclusion of the lemma does not hold. That is, there is a function $g:\calY_1 \to \mathbb{R}^{\ge 0}$ such that 
\[
\|g \|_{P_1\ell_1^{1/\beta}, 1} > \| g \|_{e^B Q_1, \beta}.
\]
In the following, we show this contradicts with $P\preceq_\beta e^B Q$. First off, for each $y_1\in \calY_1$, by the definition of $\ell_1(y_1)$, there is a function $f_{y_1}:\calY_2\to \mathbb{R}^{\ge 0}$ such that 
\[
\| f_{y_1} \|_{P_2|_{P_1=y_1}, 1} = \int f_{y_1} dP_2|_{P_1 = y_1} = \left( \int f_{y_1}^{\beta} d(\ell_1(y_1)Q_2|_{Q_1 = y_1}) \right)^{1/\beta} = \|f_{y_1}\|_{\ell_1(y_1)Q_2|_{Q_1=y_1}, \beta}.
\]
By scaling $f_{y_1}$ properly, we can ensure that $\| f_{y_1} \|_{P_2|_{P_1=y_1}, 1} = \ell_1(y_1)^{1/\beta}$. Consequently, we have
\[
\|f_{y_1}\|_{Q_2|_{Q_1=y_1}, \beta} = \|f_{y_1}\|_{\ell_1(y_1)Q_2|_{Q_1=y_1}, \beta} \cdot \ell_1(y_1)^{-1/\beta} = 1.
\]

Define a new function $f:\calY_1\times \calY_2\to \mathbb{R}^{\ge 0}$ as $f(y_1,y_2) = g(y_1)\cdot f_{y_1}(y_2)$. Then, we have
\begin{align}
\| f \|_{P,1}
&= \iint  f(y_1,y_2) dP \notag \\
&= \int \left( \int f_{y_1}(y_2) dP_2|_{P_1 = y_1} \right) g(y_1) dP_1 \notag \\
&= \int  g(y_1) d(\ell_1^{1/\beta} P_1) \notag \\
& > \left( \int g(y_1)^{\beta} d(e^B Q_1) \right)^{1/\beta} \notag \\
&=  \left( \int  g(y_1)^{\beta} \|f_{y_1}\|_{Q_2|_{Q_1=y_1}, \beta}^\beta  d(e^B Q_1) \right)^{1/\beta} \notag \\
&= \left( \int \left( g(y_1)^{\beta}  \int f_{y_1}(y_2)^{\beta} d(Q_2|_{Q_1 = y_1})   \right) d(e^B Q_1) \right)^{1/\beta} \notag \\
&= \left(\iint g(y_1)^{\beta} f_{y_1}(y_2)^{\beta}  d(e^B Q)\right)^{1/\beta} \notag \\
&= \left(\iint f^{\beta}  d(e^B Q)\right)^{1/\beta} = \| f \|_{e^B Q, \beta}.
\end{align}
This contradicts to the assumption that $P\preceq e^B Q$. Therefore, we conclude that such function $g$ does not exist and $P_1\ell_1^{1/\beta} \preceq e^B Q_1$.
\end{proof}

\subsubsection{Proof of the composition theorem}

We prove the following theorem, which is equivalent to Theorem~\ref{theo:renyi}.

\begin{theorem}\label{theo:renyi-2}
Let $\calM_1,\calM_2$ be two interactive mechanisms that run on the same data set. Suppose that $\calM_1$,$\calM_2$ are $(\alpha,\eps_1)$, $(\alpha,\eps_2)$-R\'enyi DP, respectively. Then $\COMP(\calM_1,\calM_2)$ is $(\alpha, \eps_1+\eps_2)$-R\'enyi DP.
\end{theorem}

Theorem~\ref{theo:renyi-2} implies Theorem~\ref{theo:renyi} because we can interpret $\COMP(\calM_1,\dots, \calM_k)$ as $\COMP(\COMP(\calM_1,\dots, \calM_{k-1}),\calM_k)$ and use Theorem~\ref{theo:renyi-2} inductively. Now we prove Theorem~\ref{theo:renyi-2}.

\begin{proof}
Suppose without loss of generality that both mechanisms run on a single sensitive input bit $b\in \{0,1\}$. Also suppose that there are $2T$ rounds of interactions. Starting with $\calM_1$, the adversary communicates with two mechanisms alternately. This is without loss of generality: suppose the adversary $\calA$ can decide the next query object based previous responses. Let $\mathbf{IT}(\calA:\calM_1,\calM_2)$ be the transcript of the interaction between $\calA$ and $\calM_1,\calM_2$. We reduce the interaction to a new protocol where the adversary speaks with two mechanism alternately. Let $\calA'$ denote a modification of $\calA$, defined as follows. $\calA'$ simulates $\calA$ while always alternating between two mechanisms. If the current mechanism is not the one that $\calA$ wants to speak with, $\calA'$ will send a special ``SKIP'' query, and the mechanism responds with an ``ACK'' message. After this round of interaction, $\calA'$ will switch to interact with the other mechanism, which allows it to continue simulating $\calA$. Let $\mathbf{IT}(\calA' : \calM_1,\calM_2)$ denotes the transcript of the new interaction. Therefore, for $b\in \{0,1\}$, it is easy to establish a bijection between $\supp(\mathbf{IT}(\calA' : \calM_1^b,\calM_2^b))$ and $\supp(\mathbf{IT}(\calA : \calM_1^b,\calM_2^b))$. Moreover, the bijection mapping is independent of $b$\footnote{This is to say, suppose $\calM_1^0,\calM_2^0, \calM_1^1,\calM_2^1$ are four systems, then the bijection between $\supp(\mathbf{IT}(\calA : \calM_1^b,\calM_2^b))$ and $\supp(\mathbf{IT}(\calA' : \calM_1^b,\calM_2^b))$ would be the same for $b\in \{0,1\}$.}. Therefore, bounding the divergences between $\mathbf{IT}(\calA : \calM_1^b,\calM_2^b)$, $b\in \{0,1\}$ is equivalent to bounding those between $\mathbf{IT}(\calA' : \calM_1^b,\calM_2^b)$, $b\in \{0,1\}$.

Let $\calY,\calZ$ denote the response domains of $\calM_1,\calM_2$ respectively. Also let $y_1,\dots, y_T$, $z_1,\dots, z_T$ denote the lists of responses returned by $\calM_1$ and $\calM_2$ respectively. We assume that each response $y_i,z_j$ contains a copy of the corresponding query message (so that we can recover the whole interaction history just from the responses). 

Now, fix $\calA$ to be an arbitrary adversary. Let $P,Q \in \Delta((\calY\times \calZ)^T)$ denote the output distributions when $\calA$ interacts with $(\calM^0_1,\calM^0_2)$ and $(\calM^1_1,\calM^1_2)$ respectively. Our goal is to prove that
\[
\max\left\{ D_{\alpha}(P\| Q), D_{\alpha}(Q \| P) \right\} \le \eps_1 + \eps_2.
\]
We bound $D_{\alpha}(P\| Q)$ below. The bound for $D_{\alpha}(Q\|P)$ is symmetric. 

For a distribution $D$, we always use $D(x)$ to denote $\Pr[D = x]$. Write $y = (y_1,\dots, y_T)$ where $y_i$ denotes the $i$-th response. Also write $z = (z_1,\dots, z_T)$ and denote $yz:=(y_1,z_1,\dots, y_T, z_T)$. By Lemma~\ref{lemma:renyi-charact}, it suffices to show that for every $h \colon (\calY\times \calZ)^T \to \mathbb{R}^{\ge 0}$, it holds that
\begin{align}
\sum_{y\in \calY^T,z\in\calZ^T} P(yz) h(yz) \le \left( e^{\eps_1 + \eps_2}\sum_{y\in \calY^T,z\in\calZ^T} Q(yz) h(yz)^{\beta} \right)^{1/\beta} \label{eq:renyi-goal-2}
\end{align}
where $\beta = \frac{\alpha}{\alpha - 1}$ is the H\"older conjugate of $\alpha$. 

For each $i\in [T]$, let $P^y_i, P^z_i$ be the projection of $P$ onto $y_i, z_i$. For each $i\in [T]$, let $y_{\le i}, z_{\le i}$ denote the first $i$ responses from $y$ and $z$. Denote $(yz)_{\le i} = (y_1,z_1,\dots, y_i,z_i)$. Then, let $P^y_i|_{yz_{<i}}$ denote the distribution of $y_i$ conditioning on $(yz)_{<i}$, and $P^z_i|_{yz_{<i},y_i}$ denote the distribution of $z_i$ conditioning on $(yz)_{<i}$ and $y_i$. Also define the same notation for $Q$. Then we write 
\begin{align}
\sum_{y\in \calY^T,z\in\calZ^T} P(yz) h(yz) = \sum_{(yz)_{\le T-1}} P((yz)_{\le T-1}) \sum_{y_T,z_T} P^y_T|_{yz_{<T}}(y_T)P^z_T|_{yz_{<T},y_T}(z_T) h(yz)
\label{eq:renyi-rewrite-2}
\end{align}

For every $t< T$ and every $y_{\le t}$, let $\calM^0_{1}|_{y_{\le t}}$ (resp. $\calM^{1}_1|_{y_{\le t}}$) denote the interactive system $\calM^0_1$ (resp. $\calM^1_1$) conditioning on that it has answered $y_1,\dots, y_{t}$ to the first $t$ queries. Formally, for every $(x_{t+1},y_{t+1}),\dots, (x_{t'},y_{t'})$ and $x_{t'+1}$, define
\[
\calM^b_{1}|_{y_{\le t}}((x_i,y_i)_{t<i\le t'},x_{t'+1}) := \calM^b_1((x_i,y_i)_{1\le i\le t'}, x_{t'+1}).
\]
We also define the same notation for the second mechanism $\calM_2$. Next, define
\begin{align}
\ell_{t}(y_{\le t}) := \exp\left( \sup_{A: \text{adversary}} \left\{ D_{\alpha}\big( \IT(A:\calM^0_1|_{y_{\le t}}) \| \IT(A:\calM^1_1|_{y_{\le t}}) \big) \right\} \right) \label{eq:monitor-ell}
\end{align}
and
\begin{align}
r_{t}(z_{\le t}) := \exp\left( \sup_{A: \text{adversary}} \left\{ D_{\alpha}\big( \IT(A:\calM^0_2|_{z_{\le t}}) \| \IT(A:\calM^1_2|_{z_{\le t}}) \big) \right\} \right). \label{eq:monitor-r}
\end{align}

By the assumed R\'enyi DP guarantee, we have that $\ell_0(\emptyset) \le e^{\eps_1}$ and $r_0(\emptyset) \le e^{\eps_2}$. We claim the following.

\begin{claim}\label{claim:transfer-claim}
For each $t\le T-1$ and $y_{<t},z_{<t}$, consider two measures $P^y_t|_{yz_{<t}}(y) \ell_{t}(y_{<t}\circ y)$ and $Q^y_t|_{yz_{<t}}(y)$ on the space $\calY$ (here $\circ$ denotes concatenation). It holds that 
\[
P^y_t|_{yz_{<t}}(y) \ell_{t}(y_{<t}\circ y)^{1/\beta} \preceq_{\beta} \ell_{t-1}(y_{<t}) Q^y_t|_{yz_{<t}}(y).
\]
A symmetric conclusion holds for $P^z$ and $z$. Namely
\[
P^z_t|_{yz_{<t},y_t}(z) r_{t}(z_{<t}\circ z)^{1/\beta} \preceq_{\beta} r_{t-1}(z_{<t}) Q^z_t|_{yz_{<t},y_t}(z).
\]
\end{claim}

\begin{proof}
Construct an adversary $\calA'$ interacting with $\calM^b_1|_{y_{<t}}$ as follows. $\calA'$ starts $\calA$ with the conditioning that $\calA$ has gone through the interaction history $yz_{<t}$. Then $\calA'$ simulates one step of $\calA$ and sends a query to $\calM^b_1|_{y_{<t}}$. Upon receiving the response $y$, $\calA'$ observes $y_t$ and switches to run the optimal adversary against $\calM^b_{1}|_{y_{\le t}}$ provided by \eqref{eq:monitor-ell}. By definition, $\calM^b_1|_{y_{<t}}$ is $(\alpha, \log(\ell_{t-1}(y_{<t})))$-R\'enyi DP. Applying Lemma~\ref{lemma:renyi-tracker} on $\IT(\calA', \calM^b_1|_{y_{<t}})$ completes the proof. The proof for $P^z_t$ is similar.
\end{proof}

Turning back to \eqref{eq:renyi-rewrite-2}, we first deduce that
\begin{align}
&~~~~    \sum_{(yz)_{\le T-1}} P((yz)_{\le T-1}) \sum_{y_T,z_T} P^y_T|_{yz_{<T}}(y_T)P^z_T|_{yz_{<T},y_T}(z_T) h(yz) \notag \\
&\le \sum_{(yz)_{\le T-1}} P((yz)_{\le T-1}) \sum_{y_T} P^y_T|_{yz_{<T}}(y_T) \left(  r_{T-1}(z_{<T})\sum_{z_T} Q^z_T|_{yz_{<T},y_T}(z_T) h(yz)^\beta \right)^{1/\beta} \notag \\
&\le \sum_{(yz)_{\le T-1}} P((yz)_{\le T-1}) \left( r_{T-1}(z_{<T}) \ell_{T-1}(y_{<T}) \sum_{y_T,z_T} Q^y_T|_{yz_{<T}}(y_T)  \ Q^z_T|_{yz_{<T},y_T}(z_T) h(yz)^\beta \right)^{1/\beta}. \label{eq:to-be-manipulate}
\end{align}
So far we haven't utilized Claim~\ref{claim:transfer-claim} yet. Denote 
\[
H(yz_{\le T-1}) := \left(\ell_{T-1}(y_{<T}) \sum_{y_T,z_T} Q^y_T|_{yz_{<T}}(y_T)  \ Q^z_T|_{yz_{<T},y_T}(z_T) h(yz)^\beta \right)^{1/\beta}.
\]
Applying Claim~\ref{claim:transfer-claim} on \eqref{eq:to-be-manipulate} for $P^z_{T-1}$ yields that
\begin{align}
&~~~~   \sum_{(yz)_{\le T-2}, y_{T-1}}P((yz)_{\le T-2}, y_{T-1}) \sum_{z_{T-1}} P^z_{T-1}|_{yz_{\le T-2},y_{T-1}}(z_{T-1}) r_{T-1}(z_{<T})^{1/\beta} H \notag \\
& \le  \sum_{(yz)_{\le T-2}, y_{T-1}}P((yz)_{\le T-2}, y_{T-1}) \left( r_{T-2}(z_{\le T-2}) \sum_{z_{T-1}} Q^z_{T-1}|_{yz_{\le T-2},y_{T-1}}(z_{T-1}) H^{\beta} \right)^{1/\beta}. \label{eq:manipulate-1} 
\end{align}
We proceed to apply Claim~\ref{claim:transfer-claim} on \eqref{eq:manipulate-1} for $P^y_{T-1},P^z_{T-2}, P^y_{T-2}\dots, P^z_{1}, P^y_1$ in order. We can get
\begin{align}
&~~~~    \sum_{(yz)_{\le T-1}} P((yz)_{\le T-1}) \sum_{y_T,z_T} P^y_T|_{yz_{<T}}(y_T)P^z_T|_{yz_{<T},y_T}(z_T) h(yz) \notag \\
&\le  \left( \ell_0(\emptyset) r_0(\emptyset)\sum_{yz} Q(yz) h(yz)^{\beta} \right)^{1/\beta}.
\end{align}
This shows that $P\preceq e^{\eps_1+\eps_2} Q$, which consequently implies that $D_{\alpha}(P\|Q) \le \eps_1+\eps_2$. Similarly, we can bound $D_{\alpha}(Q\|P)\le \eps_1+\eps_2$. Combining two bounds together completes the proof.
\end{proof}

\subsection{Proof for Concentrated DP}\label{append:concentrate}

In this section, we prove Corollary~\ref{coro:CDP}. We recall the definition of zero-concentrated DP and truncated concentrated DP.

\begin{definition}[zero-concentrated differential privacy, \cite{DBLP:conf/tcc/BunS16}]
Let $\rho > 0$ be a real and $\calM$ be a mechanism. $\calM$ is called $\rho$-zero-concentrated DP (or $\rho$-zCDP for short), if for every $\alpha\in (1,+\infty)$, $\calM$ is $(\alpha, \alpha\cdot \rho)$-RDP.
\end{definition}

\begin{definition}[truncated concentrated differential privacy \cite{DBLP:conf/stoc/BunDRS18}]
Let $\rho > 0, \omega > 1$ be two reals, and $\calM$ be a mechanism. $\calM$ is called $(\rho,\omega)$-truncated DP (or $(\rho,\omega)$-tCDP), if for every $\alpha \in (1,\omega)$, $\calM$ is $(\alpha, \alpha\cdot \rho)$-RDP.
\end{definition}

We are ready to prove Corollary~\ref{coro:CDP} below.

\begin{proof}
We first prove for zCDP. Suppose $\calM_1,\dots, \calM_k$ are $k$ interactive mechanisms, where for each $i\in [k]$, $\calM_i$ is $\rho_i$-zCDP. By definition, we know that $\calM_i$ is $(\alpha,\alpha\rho_i)$-RDP for every $\alpha > 1$. By Theorem~\ref{theo:renyi}, we know that $\COMP(\calM_1,\dots, \calM_k)$ satisfies $(\alpha,\alpha(\sum_{i} \rho_i))$-RDP. Since this argument holds for every $\alpha > 1$, we conclude that $\COMP(\calM_1,\dots, \calM_k)$ satisfies $(\sum_{i}\rho_i)$-zCDP.

The proof for tCDP is similar. Fix $\omega > 1$. Again let $\calM_1,\dots, \calM_k$ are $k$ interactive mechanisms, where for each $i\in [k]$, $\calM_i$ is $(\rho_i,\omega)$-tCDP. Then, for every $\alpha \in (1,\omega)$, we know that $\calM_i$ is $(\alpha,\alpha\rho_i)$-RDP by definition. Theorem~\ref{theo:renyi} then shows that $\COMP(\calM_1,\dots, \calM_k)$ satisfies $(\alpha,\alpha(\sum_{i} \rho_i))$-RDP. Since the argument holds for every $\alpha \in (1, \omega)$, we conclude that $\COMP(\calM_1,\dots, \calM_k)$ satisfies $(\sum_{i}\rho_i, \omega)$-tCDP.
\end{proof}

\section{A Motivating Example of Concurrent Composition}\label{append:example}

To demonstrate the power of concurrent composition, in this section, we use Theorem~\ref{theo:approximate} to analyze a simple private ``Guess-and-Check'' algorithm. We remark that this is a rather preliminary application: the weaker concurrent composition theorem by \citet{DBLP:conf/tcc/VadhanW21} is sufficient to do the job. However, the main purpose of this section is to highlight the importance of concurrent composition, and hopefully inspire researchers to design more sophisticated algorithms.

\medskip\noindent\textbf{Setup.} Now we describe the problem. The private algorithm holds a sensitive data set $X$. The user keeps issuing queries to the algorithm, where each query consists of a $1$-Lipschitz function $f_i$ and a guess $\tau_i\in \mathbb{R}$ for the value of $f_i(X)$. The algorithm's job is to verify if $f_i(X)\approx \tau_i$. If it is the case, the algorithm reports ``PASS'' and continues to the next query. Otherwise, the algorithm reports ``WRONG'' and a value $v_i$ that is approximately equal to $f_i(X)$ (i.e., the algorithm not only declares the invalidity of the user's guess, but also provides a correct estimation for $f_i(X)$).

We consider the following algorithm.

\newcommand{\Lap}{\mathrm{Lap}}

\begin{algorithm2e}[H]
\LinesNumbered
    \caption{The Private Guess-and-Check}
    \label{algo:guess-check}
    
    \SetKwProg{Fn}{Function}{:}{}
    \SetKwProg{Pg}{Program}{:}{\KwRet}
    
    \SetKwFunction{FSelect}{Selection}
    \SetKwFunction{FTest}{Test}
    
    \DontPrintSemicolon
    
    \SetKwFunction{FMain}{Main}
    \SetKwFunction{FQuery}{Query}
    
    \KwIn{
        Private dataset $X$. Error tolerance parameter $E > 0$. Privacy-related parameters $c\ge 1, \varepsilon \in (0, 1)$.
    }
    \Pg{}{
        $\rho\gets \Lap\left( \frac{1}{\varepsilon} \right)$   \tcp*[f]{Note that this noise has standard deviation $\approx \frac{1}{\eps}$} \;
        
        \For{$i=1,2,\dots, $}{
            Receive the next query $(f_i, \tau_i)$ \;
            $\gamma_i \gets \Lap(c/\varepsilon)$ \;
            \If{$|f_i(X)-\tau_i|+\gamma_i \ge E + \rho$}{
                $v_i\gets f_i(X) + \Lap(c/\varepsilon)$ \;
                Report $(\mathrm{WRONG}, v_i)$\;
                $t\gets t + 1$\;
                \If{$t = c$}{
                    HALT the algorithm.
                }
            }
            \Else{
                Return $\mathrm{PASS}$\;
            }
        }
    }
\end{algorithm2e}

\medskip\noindent\textbf{Discussions.} Algorithm~\ref{algo:guess-check} is parameterized by an error tolerance parameter $E > 0$ and two privacy parameters $c\ge 1, \varepsilon \in (0, 1)$. Roughly speaking, it can process queries until identifying at least $c$ queries whose guesses deviate from the true value by at least (roughly) $E$. It works by (concurrently) composing a variant of the sparse vector technique by \cite{DBLP:journals/pvldb/LyuSL17} with the standard Laplace noise-adding mechanism.

The main advantage of the \cite{DBLP:journals/pvldb/LyuSL17} SVT is that it only adds noise to the threshold once (Line 2 of algorithm~\ref{algo:guess-check}), using a \emph{much smaller} noise, which makes the SVT algorithm more accurate. Since the utility guarantee of the algorithm is not the focus of this work, we omit more discussions here and refer interested readers to \citep{DBLP:journals/pvldb/LyuSL17,DBLP:conf/nips/ZhuW20} for more detail.

We consider the privacy guarantee of Algorithm~\ref{algo:guess-check}. In fact, without the concurrent composition framework, it is not clear whether or not Algorithm~\ref{algo:guess-check} is really private! If we replace Line 7 of the algorithm by $v_i\gets 0$, then the algorithm is indeed $(3\eps, 0)$-private, because it is just a faithful implementation of the \citet{DBLP:journals/pvldb/LyuSL17} SVT. However, in Algorithm~\ref{algo:guess-check}, the algorithm reports a correct estimation $v_i$ for each inaccurate guess, which implies that the future query to the algorithm may depend on $v_i$, and thus on the private data set $X$. In this case, the original analysis from \citep{DBLP:journals/pvldb/LyuSL17} does not hold anymore.

\medskip\noindent\textbf{Analyzing the privacy.} While it is not hard to prove the privacy property of Algorithm~\ref{algo:guess-check} by examining the proof of \citet{DBLP:journals/pvldb/LyuSL17} carefully and applying some modifications, here we show that Algorithm~\ref{algo:guess-check} admits a fairly straightforward privacy proof under the concurrent composition framework, using the privacy theorem by \citet{DBLP:journals/pvldb/LyuSL17} as a black box. We do the analysis now. First, we have the following lemma from \citep{DBLP:journals/pvldb/LyuSL17}.

\begin{lemma}[Theorem~2 in \citet{DBLP:journals/pvldb/LyuSL17}]\label{lemma:Lyu-privacy}
Consider replacing Line~7 of Algorithm~\ref{algo:guess-check} with $v_i\gets 0$. The resulting algorithm is $(3\eps,0)$-DP.
\end{lemma}

The following fact is well known.

\begin{lemma}[Laplace mechanism]\label{lemma:Lap}
Consider the following algorithm: given a list of $c$ adaptively chosen, $1$-Lipschitz queries $(g_1,\cdots, g_c)$, answer each query with $g_i(X)+\Lap(c/\varepsilon)$. The algorithm is $(\eps,0)$-DP.
\end{lemma}

Combining Lemmas~\ref{lemma:Lyu-privacy} and \ref{lemma:Lap} under the concurrent composition framework directly yields the following result.

\begin{theorem}
Algorithm~\ref{algo:guess-check} is $(4\eps, 0)$-DP.
\end{theorem}
\begin{proof}
Consider simulating Algorithm~\ref{algo:guess-check} by concurrently composing two algorithms $A_1,A_2$. $A_1$ is just a modification of Algorithm~\ref{algo:guess-check} where we replace Line 7 in Algorithm~\ref{algo:guess-check} with $v_i\gets 0$. By Lemma~\ref{lemma:Lyu-privacy}, $A_1$ is $(3\eps,0)$-DP. $A_2$ accepts at most $c$ $1$-Lipschitz query. For each query $g_i$, $A_2$ responds with $g_i(X) + \Lap(c/\eps)$. By Lemma~\ref{lemma:Lap}, $A_2$ is $(\eps, 0)$-DP. By Theorem~\ref{theo:approximate}, $\COMP(A_1,A_2)$ is $(4\eps,0)$-DP.

We now describe how to simulate Algorithm~\ref{algo:guess-check} with $\COMP(A_1,A_2)$. For each query $(f_i, \tau_i)$ to Algorithm~\ref{algo:guess-check}, we first feed it into $A_1$ and observe the outcome. We pass this query if the outcome is $\mathrm{PASS}$. Otherwise, the outcome must be $(\mathrm{WRONG}, 0)$. We then query $A_2$ with $f_i$ to get an estimation $f_i+\Lap(c/\eps)$, and think of this estimation as the ``$v_i$'' returned by Algorithm~\ref{algo:guess-check}. In this way, it is easy to see that we faithfully simulate Algorithm~\ref{algo:guess-check} by interacting with $\COMP(A_1, A_2)$. Since $\COMP(A_1,A_2)$ is $(4\eps,0)$-DP, Algorithm~\ref{algo:guess-check} must be $(4\eps, 0)$-DP also. This completes the proof.
\end{proof}

\begin{remark}
Finally, we remark that a similar private ``Guess-and-Check'' algorithm was also proposed and analyzed by \citet{DBLP:conf/nips/ZhuW20}, where the authors also considered using a version of SVT \emph{without} refreshing the threshold after answering each ``meaningful'' query. Therefore, their algorithm is also subject to the concurrent composition issue, which seems to be overlooked in the original analysis of \citet{DBLP:conf/nips/ZhuW20}. Since they were working with R\'enyi DP, our Theorem~\ref{theo:renyi} provides a remedy to this issue easily.
\end{remark}

\end{document}